\documentclass[a4paper,10pt]{article}
\usepackage[utf8x]{inputenc}
\usepackage[english]{babel}
\usepackage{amssymb,amsfonts,amsmath,amsthm}
\usepackage[margin=3cm]{geometry}
\usepackage{fancyhdr}
\usepackage{hyperref}
\usepackage{mathtools}
\usepackage{booktabs}
\usepackage{pgfplots}
\pgfplotsset{compat=newest}
   \newlength\figureheight
   \newlength\figurewidth

\newtheorem{theorem}{Theorem}[section]
\newtheorem{lemma}[theorem]{Lemma}
\newtheorem{proposition}[theorem]{Proposition}
\newtheorem{corollary}[theorem]{Corollary}
\theoremstyle{definition}
\newtheorem{definition}[theorem]{Definition}
\newtheorem{assumption}[theorem]{Assumption}
\newtheorem{example}[theorem]{Example}
\theoremstyle{remark}
\newtheorem{remark}[theorem]{Remark}

\newcommand{\R}{\mathbb{R}}
\newcommand{\N}{\mathbb{N}}
\newcommand{\rmd}{\mathrm{d}}
\newcommand{\ddt}{\frac{\rmd}{\rmd t}}
\newcommand{\rme}{\mathrm{e}}
\newcommand{\hatmuR}{\hat{\mu}^R_t}
\newcommand{\hatmuE}{\hat{\mu}^E_t}
\newcommand{\hatmuC}{\hat{\mu}^C_t}
\newcommand{\hatmuF}{\hat{\mu}^F_t}
\newcommand{\gamR}{\gamma^R_t}
\newcommand{\gamE}{\gamma^E_t}
\newcommand{\gamC}{\gamma^C_t}
\newcommand{\gamF}{\gamma^F_t}
\newcommand{\ones}{\mathbf{1}_d}

\newcommand{\gam}[2]{\gamma^{#1}_{#2}}
\newcommand{\tstrut}{\rule{0pt}{2.6ex}}
\newcommand{\bstrut}{\rule[-0.9ex]{0pt}{0pt}}

\DeclareMathOperator{\E}{\mathbb{E}}

\DeclareMathOperator{\cov}{cov}
\DeclareMathOperator{\tr}{tr}

\begin{document}
 
 \title{Expert Opinions and Logarithmic Utility Maximization for Multivariate Stock Returns with Gaussian Drift}
 \author{J\"orn Sass\thanks{Department of Mathematics, University of Kaiserslautern, P.O. Box 3049, 67653 Kaiserslautern, Germany; \newline \textit{E-mail address:} sass@mathematik.uni-kl.de}, Dorothee Westphal\footnote{Department of Mathematics, University of Kaiserslautern, P.O. Box 3049, 67653 Kaiserslautern, Germany; \newline \textit{E-mail address:} westphal@mathematik.uni-kl.de} and Ralf Wunderlich\footnote{Mathematical Institute, Brandenburg University of Technology Cottbus - Senftenberg, Postfach 101344, 03013 Cottbus, Germany; \newline \textit{E-mail address:} ralf.wunderlich@b-tu.de}}
 \date{March 10, 2016}
 
 \maketitle
 
 \begin{abstract}
  This paper investigates optimal trading strategies in a financial market with multidimensional stock returns where the drift is an unobservable multivariate Ornstein-Uhlenbeck process. Information about the drift is obtained by observing stock returns and expert opinions. The latter provide unbiased estimates on the current state of the drift at discrete points in time.
  
  The optimal trading strategy of investors maximizing expected logarithmic utility of terminal wealth depends on the filter which is the conditional expectation of the drift given the available information. We state filtering equations to describe its dynamics for different information settings. Between expert opinions this is the Kalman filter. The conditional covariance matrices of the filter follow ordinary differential equations of Riccati type. We rely on basic theory about matrix Riccati equations to investigate their properties. Firstly, we consider the asymptotic behaviour of the covariance matrices for an increasing number of expert opinions on a finite time horizon. Secondly, we state conditions for the convergence of the covariance matrices on an infinite time horizon with regularly arriving expert opinions.
  
  Finally, we derive the optimal trading strategy of an investor. The optimal expected logarithmic utility of terminal wealth, the value function, is a functional of the conditional covariance matrices. Hence, our analysis of the covariance matrices allows us to deduce properties of the value function.
  
  \smallskip
  
  
 \end{abstract}

 \section{Introduction}
  
  Trading decisions in financial markets are always made based on the often very limited information on stock developments available to the investors. Such information might comprise former and present observed stock returns. Although these returns are influenced by some drift term, there is always random variation in observed data. For making trading decisions it is however of huge importance to know as much as possible about the underlying drift. Another source of information that investors often rely on when it comes to trading are expert opinions. Experts might have some deeper knowledge about the current developments in the market and are therefore able to give a more or less accurate estimate of drift terms at certain times.
  The aim of this paper is to investigate optimal portfolio trading strategies in a financial market where the drift of the stock returns is an unobserved Gaussian process. Information about the drift process is obtained from observing stock returns as well as incoming expert opinions that give an unbiased estimate of the current state of the drift at discrete points in time. An investor's objective is to find a trading strategy that maximizes expected logarithmic utility of her terminal wealth.
  
  Without expert opinions this is a classical utility maximization problem under partial information, meaning that an investor only has the information coming from observing the stock returns and cannot see the underlying stochastic drift process directly. The best estimate in a mean-square sense then is the filter. While under suitable integrability assumptions existence of optimal trading strategies can be shown, see Bj\"ork, Davis and Land\'en~\cite{bjoerk_davis_landen_2010} and  Lakner~\cite{lakner_1995}, we need models which allow for finite-dimensional filters to solve the problem completely including the computation of an optimal policy. There are essentially two cases which lead to finite-dimensional filters. Firstly, the drift process can be modeled as an Ornstein-Uhlenbeck process (OUP) as above (including the degenerate case of a static but unobserved random variable), or as a continuous time Markov chain (CTMC). The filters are the well-known Kalman and Wonham filters, respectively, see e.g.\ Elliott, Aggoun and Moore~\cite{elliott_aggoun_moore_1994}, Liptser and Shiryaev~\cite{liptser_shiryaev_1994}. In these two models the solution of the utility maximization problem is known, see Brendle~\cite{brendle_2006}, Lakner~\cite{lakner_1998}, Putsch\"ogl and Sass~\cite{putschoegl_sass_2008} and Honda~\cite{honda_2003}, Rieder and B\"auerle~\cite{rieder_baeuerle_2005}, Sass and Haussmann~\cite{sass_haussmann_2004}, respectively. 
  
  Including unbiased expert opinions reduces the variance of the filter. The better estimate then improves the expected utility. This can be seen as a continuous-time version of the static Black-Litterman approach which combines an estimate of the asset returns with expert opinions on the performance of the assets, see Black and Litterman~\cite{black_litterman_1992}.
  Frey, Gabih and Wunderlich~\cite{frey_gabih_wunderlich_2012, frey_gabih_wunderlich_2014} solve the case of an underlying CTMC with power utility, and Gabih, Kondakji, Sass and Wunderlich~\cite{gabih_kondakji_sass_wunderlich_2014} for OUP with logarithmic utility. As an approximation, also expert opinions arriving continuously in time can be introduced. This allows for more explicit solutions for the portfolio optimization problem. Davis and Lleo~\cite{davis_lleo_2013} consider this approach for an underlying OUP, Sass, Seifried and Wunderlich~\cite{sass_seifried_wunderlich_2016} address the CTMC.
  
  This paper generalizes the results from \cite{gabih_kondakji_sass_wunderlich_2014}, obtained for a market with one stock, to a financial market with $d\geqslant 1$ stocks and corresponding expert opinions. The filtering equations we derive are extensions of the one-dimensional case. The portfolio optimization results carry over to the multivariate case as well, see Theorem~\ref{thm:optimal_value}. But the convergence results of \cite{gabih_kondakji_sass_wunderlich_2014} for the conditional variance have no direct equivalents in the multivariate case, since they require and state very detailed monotonicity properties and bounds. Instead we choose suitable norms, e.g.\ the spectral norm, and here lie our main contributions. The convergence of the norms of the conditional covariances for an increasing number of expert opinions to zero can now be shown, see Theorem~\ref{thm:asymptotics_for_N_to_infinity}. The convergence of the norms for equidistant expert opinions on an infinite time horizon is more delicate, in particular when requiring monotonicity between the expert opinions which reflects the decreasing impact of the expert opinions over time. Here we state several results, showing convergence under certain conditions, see e.g.\ Theorem~\ref{thm:trace_of_gamma_E}, as well as providing counterexamples if these conditions do not hold.  
  
  
  In detail we proceed as follows. In Section~\ref{sec:market_model_and_filtering_equations} we introduce our financial market model. We assume that the drift of the stock returns is a multivariate Ornstein-Uhlenbeck process with dynamics
  \[ \rmd \mu_t = \alpha (\delta - \mu_t)\,\rmd t + \beta\,\rmd B_t, \]
  where $\alpha, \beta\in\R^{d\times d}$, $\delta\in\R^d$ and $B=(B_t)_{t\in [0,T]}$ is a $d$-dimensional Brownian motion. The drift cannot be observed by the participants in the market. Aside from the stock prices, further estimates on the current state of the drift arrive in form of expert opinions. We introduce the concept of expert opinions and specify different settings of information that is available to an investor. We assume that one investor observes stock returns only, another one only uses expert opinions for making trading decisions. A third investor is assumed to have access to both of these sources of information. The second part of Section~\ref{sec:market_model_and_filtering_equations} states the corresponding filtering equations. These give the dynamics of the filter and of the conditional covariance matrices. In the case of return observations only, the filter is the classical Kalman filter, see for example Liptser and Shiryaev~\cite{liptser_shiryaev_1994}. When we include expert opinions we make use of the discrete-time Kalman filter as described in Elliott, Aggoun and Moore~\cite{elliott_aggoun_moore_1994}.
  
  Section~\ref{sec:properties_of_the_conditional_covariance_matrix} analyzes the conditional covariance matrices. In particular, Theorem~\ref{thm:asymptotics_for_N_to_infinity} shows the limiting behaviour for an increasing number of expert opinions with some minimal reliability on a finite time horizon. This is a generalization of Proposition~4.3 from Gabih et al.~\cite{gabih_kondakji_sass_wunderlich_2014}. In Section~\ref{sec:asymptotic_results_for_an_infinite_time_horizon}, we analyze the limiting behaviour of the conditional covariance matrices on an infinite time horizon with regularly arriving expert opinions. In this context it is important to mention that the conditional covariance matrices for the investor who observes stock returns only and for the investor observing stock returns as well as expert opinions follow a matrix Riccati equation. In contrast to the one-dimensional situation this ordinary differential equation does not have a closed-form solution which makes the analysis harder. It is nevertheless possible to prove some limiting behaviour in these cases by using basic properties of Riccati differential equations as for example provided in Ku\u{c}era~\cite{kucera_1973}, Wonham~\cite{wonham_1968}, Bucy~\cite{bucy_1967} and M\aa{}rtensson~\cite{martensson_1971}.
  
  The properties of the conditional covariance matrices are helpful when looking at the portfolio optimization problem that we address in Section~\ref{sec:portfolio_optimization_problem}. We consider maximization of expected logarithmic utility of terminal wealth. The optimal strategy and value function for the different investors are computed along the lines of Gabih et al.~\cite{gabih_kondakji_sass_wunderlich_2014}. It turns out that the optimal value is a function of the corresponding conditional covariance matrices. Hence, the remaining part of the section concentrates on proving properties of the value function that can be deduced from properties of the covariance matrices.
  
  In Section~\ref{sec:numerical_results} we provide some simulations of filters and value functions that illustrate our theoretical results. We also take a short look at the concept of efficiency to analyze the value of information obtained from different sources of information.
  
  \textbf{Notation:} Throughout this paper, when considering symmetric matrices $A$ and $B$ of the same size, we will write $A\geqslant B$ or $B\leqslant A$ if the difference $A-B$ is positive semidefinite. Unless stated otherwise, whenever $A$ is a matrix, $\lVert A\rVert$ denotes the spectral norm of $A$. For a symmetric positive semidefinite matrix $A\in\R^{d\times d}$ we call a symmetric positive semidefinite matrix $B\in\R^{d\times d}$ the \emph{square root} of $A$ if $B^2=A$. The square root is unique and will be denoted by $A^{\frac{1}{2}}$.
  
 \section{Market Model and Filtering Equations}\label{sec:market_model_and_filtering_equations}
 
  \subsection{Financial Market Model}
   
   Let $T>0$ denote our finite investment horizon. We consider a filtered probability space $(\Omega,\mathcal{G},\mathbb{G},\mathbb{P})$ where the filtration $\mathbb{G}=(\mathcal{G}_t)_{t \in [0,T]}$ satisfies the usual conditions. All processes are assumed to be $\mathbb{G}$-adapted.
   In our financial market model there is one risk-free bond $S^0$ with dynamics
   \[ \rmd S^0_t = S^0_tr_t \,\rmd t, \qquad S^0_0=1. \]
   Here, $r=(r_t)_{t\in [0,T]}$ is some deterministic continuous process.
   Furthermore, the market allows investments in $d$ risky stocks $S^1, \dots, S^d$ with
   \[ \rmd S^i_t =S^i_t \biggl(\mu^i_t \,\rmd t + \sum_{j=1}^m \sigma^{ij} \,\rmd W^j_t\biggr), \qquad i=1, \dots, d,\]
   where $W=(W_t)_{t\in [0,T]}$ is an $m$-dimensional Brownian motion. We assume that the matrix $\sigma\sigma^T$ with $\sigma=(\sigma^{ij})_{i,j}$ is positive definite.
   
   Whereas the matrix $\sigma$ is constant over time, the drift process $\mu=(\mu^1, \dots, \mu^d)^T$ follows the dynamics of a multivariate Ornstein-Uhlenbeck process. More precisely,
   \[ \rmd \mu_t = \alpha (\delta - \mu_t)\,\rmd t + \beta \,\rmd B_t, \]
   where $\alpha, \beta\in\R^{d\times d}$, $\delta\in\R^d$ and $B=(B_t)_{t\in [0,T]}$ is a $d$-dimensional Brownian motion independent of $W$. The initial drift $\mu_0$ is multivariate normally distributed, $\mu_0 \sim \mathcal{N}(m_0, \Sigma_0)$, for some vector $m_0 \in\R^d$ and covariance matrix $\Sigma_0 \in\R^{d\times d}$ which is symmetric and positive semidefinite. We assume that $\mu_0$ is independent of $B$ and $W$.
   The drift process $\mu$ can be written as
   \[\mu_t = \delta + \rme^{-\alpha t}\biggl( \mu_0 - \delta + \int_0^t \rme^{\alpha s}\beta \,\rmd B_s \biggr).\]
   The mean $m_t := \E[\mu_t]$ and covariance matrix $\Sigma_t := \cov(\mu_t)$ of $\mu_t$ are given by the formulas
   \begin{align*}
    m_t &= \delta + \rme^{-\alpha t}(m_0-\delta), \\
    \Sigma_t &= \rme^{-\alpha t}\biggl( \Sigma_0 + \int_0^t \rme^{\alpha s} \beta \beta^T \rme^{\alpha^T s}\,\rmd s \biggr) \rme^{-\alpha^T t}.
   \end{align*}
   
   In our model we are interested in estimating the drift $\mu$ from observed stock prices $S^i$, $i=1, \dots, d$. Rather than working directly with the stock prices, it will prove easier to work with the stock returns $R^i$ instead, where
   \[ \rmd R^i_t = \frac{\rmd S^i_t}{S^i_t}. \]
   The return dynamics can be written as $\rmd R_t = \mu_t \,\rmd t + \sigma \,\rmd W_t$. Note that we can write the returns depending on the stock prices as
   \[ R^i_t=\log S^i_t + \sum_{j=1}^m \frac{1}{2} (\sigma^{ij})^2 t \]
   which implies that the filtration generated by the stock prices is the same as the one generated by the return processes. This is why in the following we assume that investors in the market observe stock returns instead of stock prices.
   
   In addition to observing stock returns, information on the drift process $\mu$ can be drawn from expert opinions that arrive at discrete time points and give an unbiased estimate of the drift. We model these expert opinions by fixing deterministic time points $0=t_0 < t_1 < \cdots < t_{N-1} < T$. The expert views at time $t_k$ are modeled as a random vector $Z_k = (Z^1_k, \dots, Z^d_k)^T$ with
   \[ Z_k = \mu_{t_k}+(\Gamma_k)^{\frac{1}{2}}\varepsilon_k \]
   where the matrices $\Gamma_k\in\R^{d\times d}$ are symmetric positive definite and $\varepsilon_k=(\varepsilon^1_k, \dots, \varepsilon^d_k)^T$. Here, the $\varepsilon^i_k$, $i=1, \dots, d$, $k=0, \dots, N-1$, are independent identically $\mathcal{N}(0,1)$-distributed random variables. We also assume that the $\varepsilon^i_k$ are independent from both $\mu_0$ and the Brownian motions $W$ and $B$.
   Note that $Z_k$ is multivariate $\mathcal{N}(\mu_{t_k}, \Gamma_k)$-distributed which implies that the expert opinions give an unbiased estimate of the true state of the drift at time $t_k$. The matrix $\Gamma_k$ is a means of modelling the reliability of the expert. Note that in the one-dimensional situation $\Gamma_k$ is just the variance of the expert's estimate at time $t_k$.
   
   \begin{remark}
    It is possible to allow relative expert views, meaning that an expert may also give an estimation of the difference of drifts of two stocks at time $t_k$. These relative estimations can be expressed in the form
    \[ Q_k=P_k\mu_{t_k}+\xi_k\in\R^d \]
    for some matrix $P_k\in\R^{l\times d}$, and some random variable $\xi_k$ that is multivariate normally distributed with expectation zero. Here, $l\leqslant d$ is the number of estimates an expert makes. The \emph{pick matrix} $P_k$ which we assume to have full rank contains information about which stocks are included in these estimates, see Section~3.1 of Schöttle, Werner and Zagst~\cite{schoettle_werner_zagst_2010} for a detailed description and an example.
    Note that since $P_k$ has full rank there exists some $\varphi_k\in\R^d$ such that $\xi_k=P_k\varphi_k$, for example $\varphi_k=P_k^T(P_kP_k^T)^{-1}\xi_k$. Hence,
    \[ Q_k=P_k(\mu_{t_k}+\varphi_k) \]
    where $Z_k:=\mu_{t_k}+\varphi_k$ is an absolute expert view about the state of the drift as introduced above, since $\varphi_k$ is normally distributed with expectation zero and can therefore be written as $(\Gamma_k)^{1/2}\varepsilon_k$.
   \end{remark}
   
   It remains to describe the information available to an investor. Following Gabih et al.~\cite{gabih_kondakji_sass_wunderlich_2014}, we distinguish four different investors with corresponding investor filtrations. Define $\mathbb{F}^H=(\mathcal{F}^H_t)_{t\in [0,T]}$ for $H\in\{R,E,C,F\}$. The first investor we consider can observe stock returns but not the incoming expert opinions. Therefore, her filtration $\mathcal{F}^R_t$ is for each $t\in[0,T]$ generated by the return processes $\{R^i_s \;|\; s\leqslant t, i=1,\dots, d \}$. Another investor cannot observe these stock returns or simply decides to rely on the expert opinions $Z_k$ only. Therefore, the corresponding investor filtration $\mathcal{F}^E_t$ is generated by the expert opinions $\{Z_k \;|\; t_k\leqslant t\}$. As a combination of the above filtrations, $\mathcal{F}^C_t$ is generated by $\{R^i_s \;|\; s\leqslant t, i=1, \dots, d\} \cup \{Z_k \;|\; t_k\leqslant t\}$. This filtration corresponds to an investor who has access to both stock returns and expert opinions as sources of information. For completeness, we also include an investor who can observe the drift process $\mu$ itself. In this last case of full information the investor filtration is simply given by $\mathbb{F}^F=\mathbb{G}$.
   
  \subsection{Filtering Equations}
   
   At the end of the previous subsection we have defined four investors with access to different sources of information. Only the fully informed investor can observe the drift process $(\mu_t)_{t\in[0,T]}$ directly. The other investors do not observe the drift but have to estimate it from the information available to them. Let $\mathbb{F}^H$ for $H\in\{R,E,C\}$ be the underlying investor filtration as defined in the previous subsection. In the mean-square sense, an optimal estimator for the drift $\mu_t$ at time $t$ under partial information is the conditional expectation $\hat{\mu}^H_t=\E[\mu_t | \mathcal{F}^H_t]$. These estimators are also called \emph{filters} and the aim of this subsection is to find filtering equations describing their dynamics. Furthermore, we also investigate the conditional covariance matrix
   \[ \gamma^H_t = \E\bigl[(\mu_t-\hat{\mu}^H_t)(\mu_t-\hat{\mu}^H_t)^T \bigl| \mathcal{F}^H_t\bigr] \]
   for $H\in\{R,E,C\}$ which is a measure for the distance between $\mu_t$ and its filter $\hat{\mu}^H_t$ given information $\mathcal{F}^H_t$.
   
   The investors with partial information cannot observe the drift directly. The only source of information for the first investor we consider are the return processes, meaning that $\mathbb{F}^R$ is the corresponding investor filtration.
   
   \begin{lemma}\label{lem:hatmuR_dynamics}
    The filter $\hatmuR$ follows the dynamics
    \[ \rmd\hatmuR = \alpha(\delta-\hatmuR)\,\rmd t + \gamR(\sigma\sigma^T)^{-1}(\rmd R_t-\hatmuR\,\rmd t), \]
    where $\gamR$ is the solution of the ordinary differential equation
    \[ \frac{\rmd}{\rmd t} \gamR = -\alpha \gamR -\gamR\alpha^T + \beta\beta^T - \gamR(\sigma\sigma^T)^{-1}(\gamR)^T. \]
    The initial values are $\hat{\mu}^R_0=m_0$ and $\gamma^R_0=\Sigma_0$.
   \end{lemma}
   \begin{proof}
    The dynamics follow immediately from the well-known Kalman filter, see for example Theorem~10.3 of Liptser and Shiryaev~\cite{liptser_shiryaev_1994}.
   \end{proof}
   
   Note that $\gamR$ follows an ordinary differential equation, called Riccati equation, and is hence deterministic. By definition, $\gamR$ is symmetric positive semidefinite. In the one-dimensional situation it is possible to write down a closed-form solution of the ordinary differential equation which yields an explicit form of $\gamR$, see equation~(3.3) in Gabih et al.~\cite{gabih_kondakji_sass_wunderlich_2014}. In the multidimensional case, we do not have such an explicit form of $\gamR$ in general. We will make use of basic properties of Riccati differential equations that can be found for example in Bucy~\cite{bucy_1967}, Ku\u{c}era~\cite{kucera_1973}, M\aa{}rtensson~\cite{martensson_1971} and Wonham~\cite{wonham_1968}.
   
   As a next step, we consider an investor whose filtration is $\mathbb{F}^E$, meaning that she knows the expert's opinions but does not observe the stock returns.
   
   \begin{lemma}\label{lem:hatmuE_dynamics}\mbox{}
    \begin{enumerate}
     \item[(i)] Let $t\in[0,T]$ and denote by $k$ the maximal index $j$ such that $t_j\leqslant t$. Then $t\in[t_k, t_{k+1})$ or in the case $k=N-1$ we have $t\in[t_{N-1},T]$, and it holds that
     \begin{align*}
      \hatmuE &= \rme^{-\alpha(t-t_k)}\hat{\mu}^E_{t_k} + \bigl(I_d - \rme^{-\alpha(t-t_k)}\bigr)\delta, \\
      \gamE &= \rme^{-\alpha(t-t_k)}\biggr(\gamma^E_{t_k} + \int_{t_k}^t \rme^{\alpha(s-t_k)}\beta\beta^T\rme^{\alpha^T(s-t_k)} \,\rmd s\biggr) \rme^{-\alpha^T(t-t_k)}.
     \end{align*}
     Here, $I_d$ denotes the unit matrix in $\R^{d\times d}$.
     \item[(ii)] At the information dates $t_k$, $k=0, \dots, N-1$, we get the formulas
     \begin{align*}
      \hat{\mu}^E_{t_k} &= \Lambda^E_k\hat{\mu}^E_{t_k-} + (I_d-\Lambda^E_k)Z_k, \\
      \gamma^E_{t_k} &= \Lambda^E_k\gamma^E_{t_k-}.
     \end{align*}
     Here, $\Lambda^E_k=\Gamma_k(\gamma^E_{t_k-}+\Gamma_k)^{-1}$. We set $\hat{\mu}^E_{0-}=m_0$ and $\gamma^E_{0-}= \Sigma_0$.
    \end{enumerate}
   \end{lemma}
   \begin{proof}
    (i) Note that we can write the drift $\mu_t$ at time $t$ as
    \[ \mu_t = \delta + \rme^{-\alpha(t-t_k)}\biggl(\mu_{t_k}-\delta+\int_{t_k}^t \rme^{\alpha(s-t_k)}\beta \,\rmd B_s\biggr). \]
    Also, there is no incoming information between $t_k$ and $t$, so $\mathcal{F}^E_t=\mathcal{F}^E_{t_k}$. Hence,
    \begin{align*}
     \hatmuE &= \E[\mu_t|\mathcal{F}^E_t] = \E[\mu_t|\mathcal{F}^E_{t_k}] \\
     &= \E\biggl[\delta + \rme^{-\alpha(t-t_k)}\Bigl(\mu_{t_k}-\delta+\int_{t_k}^t \rme^{\alpha(s-t_k)}\beta \,\rmd B_s\Bigr)\;\bigg|\;\mathcal{F}^E_{t_k}\biggr] \\
     &= \delta + \rme^{-\alpha(t-t_k)}\biggl(\hat{\mu}^E_{t_k}-\delta + \E\Bigl[\int_{t_k}^t \rme^{\alpha(s-t_k)}\beta \,\rmd B_s\Bigr]\biggr) \\
     &= \rme^{-\alpha(t-t_k)}\hat{\mu}^E_{t_k} + \bigl(I_d - \rme^{-\alpha(t-t_k)}\bigr)\delta,
    \end{align*}
    where we have used that the stochastic integral is independent of $\mathcal{F}^E_{t_k}$ and that it has expectation zero.
    For the conditional covariance matrix we get
    \begin{align*}
     \gamE &= \E\bigl[(\mu_t-\hatmuE)(\mu_t-\hatmuE)^T\big|\mathcal{F}^E_t\bigr] \\
     &= \E\Biggl[\biggl(\delta + \rme^{-\alpha(t-t_k)}\Bigl(\mu_{t_k}-\delta+\int_{t_k}^t \rme^{\alpha(s-t_k)}\beta \,\rmd B_s\Bigr)-\hatmuE\biggr) \\
     & \qquad\;\; \cdot\biggl(\delta + \rme^{-\alpha(t-t_k)}\Bigl(\mu_{t_k}-\delta+\int_{t_k}^t \rme^{\alpha(s-t_k)}\beta \,\rmd B_s\Bigr)-\hatmuE\biggr)^T \;\Bigg|\; \mathcal{F}^E_t\Biggr].
    \end{align*}
    When inserting the formula for $\hatmuE$ that was just proven, some terms cancel. The remaining conditional expectation can then be written as
    \begin{align*}
     & \E\biggl[\Bigl(\rme^{-\alpha(t-t_k)}(\mu_{t_k}-\hat{\mu}^E_{t_k})+\rme^{-\alpha(t-t_k)}\int_{t_k}^t \rme^{\alpha(s-t_k)}\beta \,\rmd B_s\Bigr) \\
     & \;\;\; \cdot\Bigl(\rme^{-\alpha(t-t_k)}(\mu_{t_k}-\hat{\mu}^E_{t_k})+\rme^{-\alpha(t-t_k)}\int_{t_k}^t \rme^{\alpha(s-t_k)}\beta \,\rmd B_s\Bigr)^T \;\bigg|\; \mathcal{F}^E_t\biggr].
    \end{align*}
    The expansion of this product is
    \begin{align*}
     &\E\bigl[\rme^{-\alpha(t-t_k)}(\mu_{t_k}-\hat{\mu}^E_{t_k})(\mu_{t_k}-\hat{\mu}^E_{t_k})^T \rme^{-\alpha^T(t-t_k)} \;\big|\; \mathcal{F}^E_t\bigr] \\
     & \qquad\;\; + \E\biggl[\rme^{-\alpha(t-t_k)}(\mu_{t_k}-\hat{\mu}^E_{t_k})\Bigl(\int_{t_k}^t \rme^{\alpha(s-t_k)}\beta \,\rmd B_s\Bigr)^T \rme^{-\alpha^T(t-t_k)} \;\bigg|\; \mathcal{F}^E_t\biggr] \\
     & \qquad\;\; + \E\bigl[\rme^{-\alpha(t-t_k)}\Bigl(\int_{t_k}^t \rme^{\alpha(s-t_k)}\beta \,\rmd B_s\Bigr)(\mu_{t_k}-\hat{\mu}^E_{t_k})^T \rme^{-\alpha^T(t-t_k)} \;\big|\; \mathcal{F}^E_t\bigr] \\
     & \qquad\;\; + \E\biggl[\rme^{-\alpha(t-t_k)}\Bigl(\int_{t_k}^t \rme^{\alpha(s-t_k)}\beta \,\rmd B_s\Bigr) \Bigl(\int_{t_k}^t \rme^{\alpha(s-t_k)}\beta \,\rmd B_s\Bigr)^T \rme^{-\alpha^T(t-t_k)}\;\bigg|\; \mathcal{F}^E_t\biggr] \\
     &= \rme^{-\alpha(t-t_k)}\Biggl(\gamma^E_{t_k} +  \E\biggl[\Bigl(\int_{t_k}^t \rme^{\alpha(s-t_k)}\beta \,\rmd B_s\Bigr)\Bigl(\int_{t_k}^t \rme^{\alpha(s-t_k)}\beta \,\rmd B_s\Bigr)^T\biggr]\Biggr)\rme^{-\alpha^T(t-t_k)}.
    \end{align*}
    In the last step, the mixed terms cancel because of independence. For the remaining expectation we can show that
    \[ \E\biggl[\Bigl(\int_{t_k}^t \rme^{\alpha(s-t_k)}\beta \,\rmd B_s\Bigr)\Bigl(\int_{t_k}^t \rme^{\alpha(s-t_k)}\beta \,\rmd B_s\Bigr)^T\biggr] = \int_{t_k}^t \rme^{\alpha(s-t_k)}\beta\beta^T\rme^{\alpha^T(s-t_k)}\,\rmd s \]
    and the claim follows.
    
    (ii) For the update formulas at information dates we interpret the situation as a degenerate discrete time Kalman filter with time points $t_k-$ and $t_k$. From formulas (5.12) and (5.13) in Elliott, Aggoun and Moore~\cite{elliott_aggoun_moore_1994} we get
    \begin{align*}
     \hat{\mu}^E_{t_k} &= \hat{\mu}^E_{t_k-} + \gamma^E_{t_k-}(\gamma^E_{t_k-}+\Gamma_k)^{-1}(Z_k-\hat{\mu}^E_{t_k-}) \\
     &= \bigl(I_d-\gamma^E_{t_k-}(\gamma^E_{t_k-}+\Gamma_k)^{-1}\bigr)\hat{\mu}^E_{t_k-} + \gamma^E_{t_k-}(\gamma^E_{t_k-}+\Gamma_k)^{-1}Z_k \\
     &= \Lambda^E_k\hat{\mu}^E_{t_k-} + (I_d-\Lambda^E_k)Z_k
    \end{align*}
    for the conditional expectation and
    \begin{align*}
     \gamma^E_{t_k} &= \E[(\mu_{t_k}-\hat{\mu}^E_{t_k})(\mu_{t_k}-\hat{\mu}^E_{t_k})^T | \mathcal{F}^E_{t_k}] \\
     &= \gamma^E_{t_k-}-\gamma^E_{t_k-}(\gamma^E_{t_k-}+\Gamma_k)^{-1}\gamma^E_{t_k-} \\
     &= \bigl(I_d-\gamma^E_{t_k-}(\gamma^E_{t_k-}+\Gamma_k)^{-1}\bigr)\gamma^E_{t_k-} \\
     &= \Lambda^E_k\gamma^E_{t_k-}
    \end{align*}
    for the conditional covariance matrix. These are the update formulas for the filter and the conditional covariance matrices at information dates. Alternatively, we can also compute the estimator $\hat{\mu}^E_{t_k}$ and its conditional covariance matrix as a Bayesian update of $\hat{\mu}^E_{t_k-}$ given the $\mathcal{N}(\mu_{t_k},\Gamma_k)$-distributed expert opinion $Z_k$, see for example Theorem~II.8.2 in Shiryaev~\cite{shiryaev_1996}.
   \end{proof}
   
   From the second part of the previous lemma one sees that at the information dates $t_k$ the filter $\hat{\mu}^E_{t_k}$ is a weighted mean of the filter $\hat{\mu}^E_{t_k-}$ before the update and the expert opinion $Z_k$. The corresponding weights depend on the matrix $\Gamma_k$ which is the covariance matrix of the expert opinion.
   
   \begin{proposition}\label{prop:update_decreases_covariance}
    For fixed $k\in\{0,\dots,N-1\}$ it holds $\gamma^E_{t_k} \leqslant \Gamma_k$ and $\gamma^E_{t_k} \leqslant \gamma^E_{t_k-}$.
   \end{proposition}
   \begin{proof}
    Using the update formula for $\gamma^E_{t_k}$ and expanding one term by $\Gamma_k$ we get the representation
    \begin{align*}
     \gamma^E_{t_k} &= \Gamma_k(\gamma^E_{t_k-}+\Gamma_k)^{-1}\gamma^E_{t_k-} \\
     &= \Gamma_k(\gamma^E_{t_k-}+\Gamma_k)^{-1}(\gamma^E_{t_k-}+\Gamma_k -\Gamma_k) \\
     &= \Gamma_k - \Gamma_k(\gamma^E_{t_k-}+\Gamma_k)^{-1}\Gamma_k.
    \end{align*}    
    Since $(\gamma^E_{t_k-}+\Gamma_k)^{-1}$ is symmetric positive definite there exists some matrix $A_k$ such that $(\gamma^E_{t_k-}+\Gamma_k)^{-1}=A_kA_k^T$. Then
    \[ \Gamma_k(\gamma^E_{t_k-}+\Gamma_k)^{-1}\Gamma_k = \Gamma_kA_kA_k^T\Gamma_k = \Gamma_kA_k(\Gamma_kA_k)^T \]
    by symmetry of $\Gamma_k$. Hence, $\Gamma_k(\gamma^E_{t_k-}+\Gamma_k)^{-1}\Gamma_k$ is symmetric positive semidefinite which yields $\gamma^E_{t_k} \leqslant \Gamma_k$.
    Likewise, when adding and subtracting $\gamma^E_{t_k-}$ instead,
    \begin{align*}
     \gamma^E_{t_k} &= \Gamma_k(\gamma^E_{t_k-}+\Gamma_k)^{-1}\gamma^E_{t_k-} \\
     &= (\Gamma_k+\gamma^E_{t_k-}-\gamma^E_{t_k-})(\gamma^E_{t_k-}+\Gamma_k)^{-1}\gamma^E_{t_k-} \\
     &= \gamma^E_{t_k-} - \gamma^E_{t_k-}(\gamma^E_{t_k-}+\Gamma_k)^{-1}\gamma^E_{t_k-}.
    \end{align*}
    As above, we can also show that $\gamma^E_{t_k-}(\gamma^E_{t_k-}+\Gamma_k)^{-1}\gamma^E_{t_k-}$ is positive semidefinite, hence $\gamma^E_{t_k} \leqslant \gamma^E_{t_k-}$.
   \end{proof}
   
   So far, we have considered $\mathbb{F}^R$ and $\mathbb{F}^E$ as investor filtrations. A rational investor in a market will however use all available information. So the case that we are most interested in is the investor filtration $\mathbb{F}^C$ which includes return observations as well as expert opinions. The formulas for the filter $\hatmuC$ and the conditional covariance matrices $\gamC$ can be deduced similarly to the cases of only return observations or only expert opinions.
   
   \begin{lemma}\label{lem:hatmuC_dynamics}\mbox{}
    \begin{enumerate}
     \item[(i)] Let $t\in[0,T]$ and denote by $k$ the maximal index $j$ such that $t_j\leqslant t$ under the convention that $t_N=T$. Then for $t\in[t_k,t_{k+1})$ it holds
     \[ \rmd \hatmuC = \alpha (\delta-\hatmuC)\,\rmd t +\gamC (\sigma\sigma^T)^{-1}(\rmd R_t-\hatmuC \,\rmd t) \]
     where $\gamC$ follows the ordinary differential equation
     \[ \frac{\rmd}{\rmd t} \gamC = -\alpha\gamC -\gamC\alpha^T +\beta\beta^T - \gamC(\sigma\sigma^T)^{-1}(\gamC)^T. \]
     The initial values are $\hat{\mu}^C_{t_k}$ and $\gamma^C_{t_k}$, respectively.
     \item[(ii)] The update formulas at information dates $t_k$ are
     \begin{align*}
      \hat{\mu}^C_{t_k} &= \Lambda^C_k\hat{\mu}^C_{t_k-}+(I_d-\Lambda^C_k)Z_k, \\
      \gamma^C_{t_k} &= \Lambda^C_k\gamma^C_{t_k-},
     \end{align*}
     where $\Lambda^C_k=\Gamma_k(\gamma^C_{t_k-}+\Gamma_k)^{-1}$. Here, we set $\hat{\mu}^C_{0-}=m_0$ and $\gamma^C_{0-}=\Sigma_0$.
    \end{enumerate}
   \end{lemma}
   \begin{proof}
    (i) Between two information dates, no additional expert opinions arrive. Hence, only return observations contribute to the filtration, meaning that $\mathcal{F}^C_t=\mathcal{F}^C_{t_k}\vee\sigma(R_s \;|\; t_k<s\leqslant t)$. Therefore, in $[t_k,t_{k+1})$, $k=0, \dots, N-2$, and in $[t_{N-1},T]$ we are in the standard situation of the Kalman filter. The dynamics follow as in Lemma~\ref{lem:hatmuR_dynamics}.
    
    (ii) At the information dates $t_k$ we use, as in the proof of Lemma~\ref{lem:hatmuE_dynamics}, the degenerate discrete time Kalman filter or a Bayesian update formula.
   \end{proof}
   
   The result from Proposition~\ref{prop:update_decreases_covariance} can also be stated in an analogue way for the investor who observes stock returns as well as expert opinions.
   
   \begin{proposition}\label{prop:update_decreases_covariance_C}
    For fixed $k\in\{0,\dots,N-1\}$ it holds $\gamma^C_{t_k}\leqslant\Gamma_k$ and $\gamma^C_{t_k}\leqslant\gamma^C_{t_k-}$.
   \end{proposition}
   \begin{proof}
    The proof uses the update formulas from Lemma~\ref{lem:hatmuC_dynamics} and works analogously to the proof of Proposition~\ref{prop:update_decreases_covariance}.
   \end{proof}
   
   For the sake of completeness we consider as a last case the situation of full information, i.e.\ where the investor filtration is $\mathbb{F}^F=\mathbb{G}$. This case corresponds to an investor who is able to observe the drift process directly. This situation will not occur in practice. We consider it as a reference case however to compare it to the other settings of information. It is clear that in this situation $\hatmuF=\E[\mu_t|\mathcal{F}^F_t]=\mu_t$ and $\gamF=\E[(\mu_t-\hatmuF)(\mu_t-\hatmuF)^T|\mathcal{F}^F_t]=\mathbf{0}_d$ for all $t\in[0,T]$. Here, $\mathbf{0}_d$ denotes the zero matrix in $\R^{d\times d}$.
   
 \section{Properties of the Conditional Covariance Matrix}\label{sec:properties_of_the_conditional_covariance_matrix}
  
  We have seen that for any of the cases $H\in\{R,E,C,F\}$ the conditional covariance matrix of the filter, $\gamma^H_t$, is deterministic. Since it gives information about the quality of the filter as an estimator for the drift, we are interested in stating some properties of $\gamma^H_t$.
  
  \begin{assumption}
   In Sections~\ref{sec:properties_of_the_conditional_covariance_matrix} and \ref{sec:asymptotic_results_for_an_infinite_time_horizon}, we assume that $\alpha$ is a symmetric positive definite matrix and that $\beta\beta^T$ is also positive definite.
  \end{assumption}
  
  For $H\in\{R,E,C,F\}$ one can easily prove that
  \begin{equation}\label{eq:second_moments_filter}
   \E[\hat{\mu}^H_t(\hat{\mu}^H_t)^T]=\E[\mu_t\mu_t^T]-\gamma^H_t = \Sigma_t+m_tm_t^T-\gamma^H_t
  \end{equation}
  for $t\in[0,T]$. This equality will be useful for connecting the filter with its covariance matrix.
  
  \subsection{Comparison of Different Investors}
   
   First, we compare the covariance matrix of an investor who observes both returns and expert opinions with that of an investor who has access to only one of these sources of information. It can be expected that the additional information yields a more precise estimate of the drift $\mu_t$.
   
   \begin{proposition}\label{prop:comparison_of_gammas}
    For all $t\in[0,T]$ we have the inequalities $\gamC\leqslant\gamR$ and $\gamC\leqslant\gamE$.
   \end{proposition}
   \begin{proof}
    Fix some $x\in\R^d$. We use the fact that for any random variable $X$ and $\sigma$-algebra $\mathcal{H}$ the conditional expectation $\E[X|\mathcal{H}]$ is the best mean-square estimate for $X$, meaning that
    \[ \E\bigl[(X-\E[X|\mathcal{H}])^2\bigr]\leqslant\E\bigl[(X-Y)^2\bigr] \]
    for all $\mathcal{H}$-measurable random variables $Y$. Now,
    \begin{align*}
     x^T\gamC x &= x^T\E\bigl[(\mu_t-\hatmuC)(\mu_t-\hatmuC)^T\;\big|\;\mathcal{F}^C_t\bigr]x \\
     &= \E\bigl[x^T(\mu_t-\hatmuC)(\mu_t-\hatmuC)^Tx\;\big|\;\mathcal{F}^C_t\bigr] \\
     &= \E\bigl[\bigl(x^T(\mu_t-\hatmuC)\bigr)^2\;\big|\;\mathcal{F}^C_t\bigr] \\
     &= \E\bigl[\bigl(x^T\mu_t-\E[x^T\mu_t|\mathcal{F}^C_t]\bigr)^2\;\big|\;\mathcal{F}^C_t\bigr].
    \end{align*}
    Since $\mathcal{F}^R_t \subseteq \mathcal{F}^C_t$ for all $t\geqslant 0$, it follows
    \[ \E[x^T\gamC x] = \E\bigl[\bigl(x^T\mu_t-\E[x^T\mu_t|\mathcal{F}^C_t]\bigr)^2\bigr] \leqslant \E\bigl[\bigl(x^T\mu_t-\E[x^T\mu_t|\mathcal{F}^R_t]\bigr)^2\bigr] = \E[x^T\gamR x]. \]
    We already know that $\gamC$ and $\gamR$ are deterministic, hence $x^T\gamC x\leqslant x^T\gamR x$. The proof of the second inequality $x^T\gamC x\leqslant x^T\gamE x$ goes completely analogously.
   \end{proof}
   
   Now that we have derived the filtering equations for the different investors in the market and stated some first properties of the conditional covariance matrices, we take a short look at the dynamics of $\gamma^H_t$ for $H\in\{R,E,C\}$ in an example.
   
   \begin{example}\label{ex:three_gammas}
    We assume that we have an investment horizon $T$ of one year and equidistant expert opinions each month which corresponds to setting $N=12$. We consider a financial market with $d=3$ stocks and $m=d$. The model parameters for the drift dynamics are
    \[ \alpha=\begin{pmatrix*}[r] 0.11 & -0.48 & 0.65 \\ -0.48 & 2.28 & -3.06 \\ 0.65 & -3.06 & 4.18 \end{pmatrix*} \quad\text{ and }\quad \beta=\begin{pmatrix*}[r] 0.87 & -0.53 & -0.22 \\ -0.53 & 0.87 & -0.02 \\ -0.22 & -0.02 & 0.29 \end{pmatrix*}, \]
    and the matrices
    \[ \sigma=\begin{pmatrix*}[r] \phantom{-}0.09 & -0.13 & 0.16 \\ 0.14 & 0.03 & -0.17 \\ 0.05 & -0.13 & -0.06 \end{pmatrix*} \quad\text{ and }\quad \Sigma_0=\begin{pmatrix*}[r] \phantom{-}0.16 & 0.12 & 0.01 \\ 0.12 & 0.19 & -0.04 \\ 0.01 & -0.04 & 0.27 \end{pmatrix*} \]
    are the volatility matrix of the returns and the covariance matrix of $\mu_0$, respectively. The expert's reliability is given by the covariance matrices
    \[ \Gamma_k=\begin{pmatrix*}[r] \phantom{-}1.14 & 0.15 & 0.58 \\ 0.15 & 1.67 & -0.73 \\ 0.58 & -0.73 & 2.67 \end{pmatrix*} \]
    for each $k=0, \dots, N-1$, in particular the covariance matrix of the expert's estimates does not depend on the current time point.
    
    In Figure~\ref{fig:three_gammas} the spectral norms of $\gamma^R_t$, $\gamma^E_t$ and $\gamma^C_t$ are plotted against time for the parameters defined above. For the investor who observes stock returns only, one can see that the spectral norm of $\gamR$ starts in $\lVert\Sigma_0\rVert$ and seems to converge to some value for increasing $t$. Note that the mapping $t\mapsto\lVert\gamR\rVert$ is not monotone, other than in the one-dimensional situation.
    When looking at the investor who observes expert opinions only, one realizes that at each information date, the norm of $\gamE$ decreases. This is due to what we have shown in Proposition~\ref{prop:update_decreases_covariance}. For large $k$ we see that the norm of $\gamE$ increases between information dates $t_k$ and $t_{k+1}$. The norms of $\gamma^E_{t_k}$ and of $\gamma^E_{t_k-}$ approximate some finite value.
    For the investor who observes stock returns as well as expert opinions we see that the norm of $\gamC$ always lies below the minimum of the norm of $\gamR$ and the norm of $\gamE$. As in the case for expert opinions only, the norm decreases at each information date, see Proposition~\ref{prop:update_decreases_covariance_C}. Also, the norms of $\gamma^C_{t_k}$ and of $\gamma^C_{t_k-}$ seem to converge, and for all $k$ large enough the norm is strictly increasing in between information dates $t_k$ and $t_{k+1}$.
   \end{example}
   
   \begin{figure}[ht]
    \centering
    \setlength\figureheight{6cm}
    \setlength\figurewidth{0.8\textwidth}
    \input{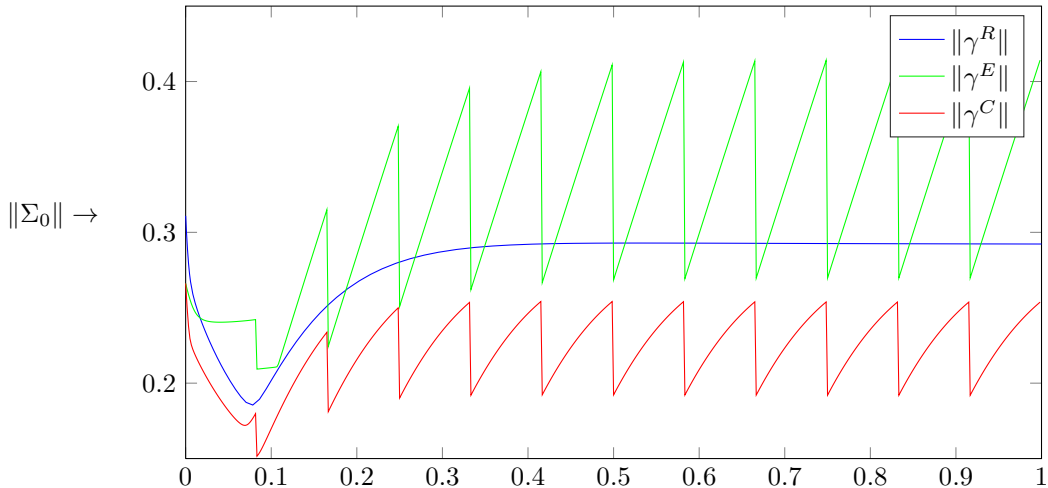}
    \caption{Development of $\lVert\gamma^H_t\rVert$ for $t\in[0,T]$, $H\in\{R,E,C\}$, in Example~\ref{ex:three_gammas}}
    \label{fig:three_gammas}
   \end{figure}
   
  \subsection{Asymptotics for an Increasing Number of Expert Opinions}
   
   We now address the question what happens when the number of dates at which expert opinions arrive goes to infinity. It stands to reason that when increasing the number of expert opinions such that the time between any two information dates goes to zero, we get an arbitrarily accurate estimate of the drift process $\mu$, at least when we assume a minimal level of reliability of the experts. The corresponding statement in a financial market with one stock is proven in Proposition~4.3 from Gabih et al.~\cite{gabih_kondakji_sass_wunderlich_2014}. The result in a market with $d$ stocks is formalized in the following theorem.
   
   \begin{theorem}\label{thm:asymptotics_for_N_to_infinity}
    Let $0=t_0^{(N)}<t_1^{(N)}<\cdots<t_{N-1}^{(N)}<T$ be a sequence of partitions of the interval $[0,T]$. To shorten notation, we will write $t_N^{(N)}=T$ for all $N$. Assume that for the mesh size
    \[ \Delta_N=\max_{k=1, \dots, N} \bigl(t_k^{(N)}-t_{k-1}^{(N)}\bigr) \]
    we have $\lim_{N\to \infty} \Delta_N=0$. Denote by $\Gamma_k^{(N)}$, $k=0, \dots, N-1$, the covariance matrices of the expert opinions at time $t_k^{(N)}$, and assume that there exists some $C>0$ such that for all $N\in\N$, $k=0, \dots, N-1$, it holds $\lVert\Gamma_k^{(N)}\rVert \leqslant C$ and that $\Gamma_0^{(N)}=\Gamma_0$ does not depend on $N$.
    
    Then for all $u\in(0,T]$ the conditional covariance matrices $\gamma_u^{E,N}$ and $\gamma_u^{C,N}$ that correspond to these $N$ expert opinions fulfill
    \[ \lim_{N\to\infty} \bigl\lVert \gamma_u^{E,N} \bigr\rVert = \lim_{N\to\infty} \bigl\lVert \gamma_u^{C,N} \bigr\rVert = 0. \]
   \end{theorem}
   \begin{proof}
    Throughout the proof we write $\lambda_{\max}(A)$ and $\lambda_{\min}(A)$ for the maximal and minimal eigenvalue of an arbitrary symmetric matrix $A$. This is well-defined since all eigenvalues of a symmetric matrix are real-valued. Furthermore, since $\lVert A \rVert$ is the square root of the maximal eigenvalue of $A^TA$, we can conclude that for symmetric positive semidefinite matrices $A$ it holds $\lVert A \rVert = \lambda_{\max}(A)$.
    
    First, we note that
    \[ \bigl\lVert \gamma_u^{C,N} \bigr\rVert = \lambda_{\max}(\gamma_u^{C,N}) = \frac{v^T\gamma_u^{C,N}v}{v^Tv} \]
    for an eigenvector $v$ of $\gamma_u^{C,N}$ to the eigenvalue $\lambda_{\max}(\gamma_u^{C,N})$. Now, by Proposition~\ref{prop:comparison_of_gammas}
    \[ \frac{v^T\gamma_u^{C,N}v}{v^Tv} \leqslant \frac{v^T\gamma_u^{E,N}v}{v^Tv} \leqslant \max_{x\in \R^d, x\neq0} \frac{x^T\gamma_u^{E,N}x}{x^Tx} = \lambda_{\max}(\gamma_u^{E,N}) = \bigl\lVert \gamma_u^{E,N} \bigr\rVert. \]
    Hence, it suffices to prove the claim for $\gamma_u^{E,N}$. To shorten notation we write $\gamma_u^N$ for $\gamma_u^{E,N}$ in the following. We also write $t_k$ for time points $t_k^{(N)}$, keeping the dependency on $N$ in mind.
    
    Let $N\in\N$ and $k\in\{0,\dots,N-1\}$. For any $t\in[t_k,t_{k+1})$ we have shown in Lemma~\ref{lem:hatmuE_dynamics} that
    \begin{equation}\label{eq:dynamics_of_gamma_E_N}
     \gamma_t^N = \rme^{-\alpha(t-t_k)}\gamma_{t_k}^N\rme^{-\alpha(t-t_k)} + \int_{t_k}^t \rme^{-\alpha(t-s)}\beta\beta^T \rme^{-\alpha(t-s)} \,\rmd s.
    \end{equation}
    Recall that we assume $\alpha$ to be a symmetric positive definite matrix.
    At the information dates the update is given by
    \[ \gamma_{t_k}^N = \Lambda_k^N\gamma_{t_k-}^N, \quad \text{ where } \quad \Lambda_k^N = \Gamma_k^{(N)}\bigl(\gamma_{t_k-}^N + \Gamma_k^{(N)}\bigr)^{-1}. \]
    The spectral norm of the first summand in \eqref{eq:dynamics_of_gamma_E_N} fulfills due to submultiplicativity
    \begin{equation}\label{eq:gamma_E_N_first_summand}
     \bigl\lVert \rme^{-\alpha(t-t_k)}\gamma_{t_k}^N\rme^{-\alpha(t-t_k)} \bigr\rVert \leqslant \bigl\lVert \rme^{-\alpha(t-t_k)} \bigr\rVert \,\bigl\lVert \gamma_{t_k}^N \bigr\rVert \,\bigl\lVert \rme^{-\alpha(t-t_k)} \bigr\rVert.
    \end{equation}
    Now since $\alpha$ is symmetric positive definite, and for the spectrum of a matrix exponential it holds $\sigma(\rme^\alpha)=\{\rme^\lambda \;|\; \lambda\in\sigma(\alpha)\}$, we can conclude that $\rme^\alpha$ is also symmetric positive definite. Hence,
    \[ \bigl\lVert \rme^{-\alpha(t-t_k)} \bigr\rVert = \frac{1}{\lambda_{\min}(\rme^{\alpha(t-t_k)})} = \frac{1}{\min_{\lambda\in\sigma(\alpha)} \rme^{\lambda(t-t_k)}} = \frac{1}{\rme^{\lambda_{\min}(\alpha)(t-t_k)}} \leqslant 1. \]
    Combining this with \eqref{eq:gamma_E_N_first_summand} yields
    \begin{equation}\label{eq:gamma_E_N_first_summand_2}
     \bigl\lVert \rme^{-\alpha(t-t_k)}\gamma_{t_k}^N\rme^{-\alpha(t-t_k)} \bigr\rVert \leqslant \bigl\lVert \gamma_{t_k}^N \bigr\rVert.
    \end{equation}
    By the same argument, we can conclude for the norm of the second summand in \eqref{eq:dynamics_of_gamma_E_N} that
    \begin{align*}
     \biggl\lVert \int_{t_k}^t \rme^{-\alpha(t-s)}\beta\beta^T \rme^{-\alpha(t-s)} \,\rmd s \biggr\rVert &\leqslant \int_{t_k}^t \bigl\lVert \rme^{-\alpha(t-s)} \bigr\rVert \,\bigl\lVert \beta\beta^T \bigr\rVert\, \bigl\lVert \rme^{-\alpha(t-s)} \bigr\rVert \,\rmd s \\     
     &\leqslant \bigl\lVert \beta\beta^T \bigr\rVert (t-t_k) \leqslant \bigl\lVert \beta\beta^T \bigr\rVert \Delta_N.
    \end{align*}
    This, together with \eqref{eq:gamma_E_N_first_summand_2}, yields for any $t\in(0,T]$ with $t\in[t_k,t_{k+1})$ that
    \begin{equation}\label{eq:gamma_E_N_inequality}
     \bigl\lVert \gamma_t^N \bigr\rVert \leqslant \bigl\lVert \gamma_{t_k}^N \bigr\rVert + \Delta_N \bigl\lVert \beta\beta^T \bigr\rVert = \bigl\lVert \Lambda_k^N\gamma_{t_k-}^N \bigr\rVert + \Delta_N \bigl\lVert \beta\beta^T \bigr\rVert.
    \end{equation}
    Note that since $\beta\beta^T$ is positive definite the matrices $\gamma_{t_k-}^N$ are invertible for all $k\geqslant 1$. By our assumption on the mesh size we can conclude for any $t\in(0,T]$ that $t\geqslant t_1=t^{(N)}_1$ for all $N$ large enough. The first summand in \eqref{eq:gamma_E_N_inequality} can then be written as
    \begin{align*}
     \bigl\lVert \Lambda_k^N\gamma_{t_k-}^N \bigr\rVert &= \Bigl\lVert \Gamma_k^{(N)}\bigl(\gamma_{t_k-}^N + \Gamma_k^{(N)}\bigr)^{-1} \gamma_{t_k-}^N \Bigr\rVert \\
     &= \Bigl\lVert \Gamma_k^{(N)} \Bigl[\gamma_{t_k-}^N \bigl(I_d + (\gamma_{t_k-}^N)^{-1}\Gamma_k^{(N)}\bigr)\Bigr]^{-1} \gamma_{t_k-}^N \Bigr\rVert \\
     &= \Bigl\lVert \Gamma_k^{(N)} \bigl(I_d + (\gamma_{t_k-}^N)^{-1}\Gamma_k^{(N)}\bigr)^{-1} \Bigr\rVert \\
     &= \Bigl\lVert \Gamma_k^{(N)} \Bigl[\bigl((\Gamma_k^{(N)})^{-1}+(\gamma_{t_k-}^N)^{-1}\bigr)\Gamma_k^{(N)}\Bigr]^{-1} \Bigr\rVert \\
     &= \Bigl\lVert \bigl( (\Gamma_k^{(N)})^{-1}+(\gamma_{t_k-}^N)^{-1} \bigr)^{-1} \Bigr\rVert \\
     &= \frac{1}{\lambda_{\min}\bigl((\Gamma_k^{(N)})^{-1}+(\gamma_{t_k-}^N)^{-1}\bigr)}.
    \end{align*}
    Weyl's theorem, see for example Theorem~4.3.1 in Horn and Johnson~\cite{horn_johnson_1990}, states that for any symmetric matrices $A$ and $B$ we have the inequality $\lambda_{\min}(A+B) \geqslant \lambda_{\min}(A) + \lambda_{\min}(B)$.
    This implies that
    \begin{align*}
     \frac{1}{\lambda_{\min}\bigl((\Gamma_k^{(N)})^{-1}+(\gamma_{t_k-}^N)^{-1}\bigr)} &\leqslant \frac{1}{\lambda_{\min}\bigl((\Gamma_k^{(N)})^{-1}\bigr)+\lambda_{\min}\bigl((\gamma_{t_k-}^N)^{-1}\bigr)} \\
     &= \frac{1}{\frac{1}{\lVert\Gamma_k^{(N)}\rVert}+\frac{1}{\lVert\gamma_{t_k-}^N\rVert}} = \frac{\bigl\lVert\Gamma_k^{(N)}\bigr\rVert \bigl\lVert\gamma_{t_k-}^N\bigr\rVert}{\bigl\lVert\Gamma_k^{(N)}\bigr\rVert + \bigl\lVert\gamma_{t_k-}^N\bigr\rVert} \\
     &\leqslant \biggl(\frac{C}{C + \bigl\lVert\gamma_{t_k-}^N\bigr\rVert}\biggr) \bigl\lVert\gamma_{t_k-}^N\bigr\rVert,
    \end{align*}
    where we have used that $\lVert\Gamma_k^{(N)}\rVert\leqslant C$.
    Inserting this into \eqref{eq:gamma_E_N_inequality}, we get
    \begin{equation}\label{eq:iterate_this_inequality}
     \bigl\lVert \gamma_t^N \bigr\rVert \leqslant \biggl(\frac{C}{C + \bigl\lVert\gamma_{t_k-}^N\bigr\rVert}\biggr) \bigl\lVert\gamma_{t_k-}^N\bigr\rVert + \Delta_N \bigl\lVert \beta\beta^T \bigr\rVert.
    \end{equation}
    Next, we iterate \eqref{eq:iterate_this_inequality} to get
    \[ \bigl\lVert \gamma_t^N \bigr\rVert \leqslant \prod_{j=1}^k \biggl(\frac{C}{C + \bigl\lVert\gamma_{t_j-}^N\bigr\rVert}\biggr) \bigl\lVert \gamma^N_0 \bigr\rVert + \Delta_N \bigl\lVert \beta\beta^T \bigr\rVert \sum_{j=0}^k \prod_{l=1}^j \biggl(\frac{C}{C + \bigl\lVert\gamma_{t_{k+1-l}-}^N\bigr\rVert}\biggr). \]
    Setting $L_k^N = \max_{j=1, \dots, k} \Bigl(\frac{C}{C + \lVert\gamma_{t_j-}^N\rVert}\Bigr)$, we conclude
    \begin{equation}\label{eq:gamma_E_N_inequality_2}
     \bigl\lVert \gamma_t^N \bigr\rVert \leqslant (L_k^N)^{k} \bigl\lVert \gamma^N_0 \bigr\rVert + \Delta_N \bigl\lVert \beta\beta^T \bigr\rVert \sum_{j=0}^k (L_k^N)^j.
    \end{equation}
    Now let $u\in(0,T]$ and $\varepsilon>0$. For all $N\in\N$ let $k_N$ denote the index for which $u\in[t_{k_N},t_{k_N+1})$, or, in the case $u=T$, let $k_N=N$. Suppose that for all $N_0\in\N$ there is some $N\geqslant N_0$ such that
    \[ \bigl\lVert \gamma_{t_1-}^N \bigr\rVert, \dots, \bigl\lVert \gamma_{t_{k_N}-}^N \bigr\rVert \geqslant \varepsilon/2. \]
    Then for all $j=1, \dots, k_N$ it holds
    \[ \frac{C}{C + \bigl\lVert\gamma_{t_j-}^N\bigr\rVert} \leqslant \frac{C}{C+\varepsilon/2}, \quad \text{ hence } \quad L_{k_N}^N \leqslant \frac{C}{C+\varepsilon/2}. \]
    Now, equation \eqref{eq:gamma_E_N_inequality_2} implies
    \begin{align*}
     \bigl\lVert \gamma_{t_{k_N}-}^N \bigr\rVert &\leqslant \biggl(\frac{C}{C+\varepsilon/2}\biggr)^{k_N-1} \bigl\lVert \gamma^N_0 \bigr\rVert + \Delta_N \bigl\lVert \beta\beta^T \bigr\rVert \sum_{j=0}^{k_N-1} \biggl(\frac{C}{C+\varepsilon/2}\biggr)^j \\
     &\leqslant \biggl(\frac{C}{C+\varepsilon/2}\biggr)^{k_N-1} \bigl\lVert \gamma^N_0 \bigr\rVert + \Delta_N \bigl\lVert \beta\beta^T \bigr\rVert \frac{2C+\varepsilon}{\varepsilon}.
    \end{align*}
    Since our assumption on the mesh size implies $\lim_{N\to\infty} k_N = \infty$ and $\gamma^N_0=\Gamma_0(\Sigma_0+\Gamma_0)^{-1}\Sigma_0$ does not depend on $N$, the right-hand side of this inequality goes to zero when $N$ tends to infinity. So there is some $N_0\in\N$ such that for all $N\geqslant N_0$ it holds $\lVert \gamma_{t_{k_N}-}^N \rVert < \varepsilon/2$. This is a contradiction to our assumption.
     
    Hence, there is some $N_0\in\N$ such that for all $N\geqslant N_0$ there exists some index $1\leqslant l_N\leqslant k_N$ with $\lVert \gamma_{t_{l_N}-}^N \rVert < \varepsilon/2$. We denote by $l_N$ the maximal index less or equal $k_N$ with that property.
    If $l_N=k_N$, then
    \[ \bigl\lVert \gamma_u^N \bigr\rVert \leqslant \bigl\lVert \Lambda_{k_N}^N \gamma_{t_{k_N}-}^N \bigr\rVert + \Delta_N \bigl\lVert \beta\beta^T \bigr\rVert \leqslant \bigl\lVert \gamma_{t_{k_N}-}^N \bigr\rVert + \Delta_N \bigl\lVert \beta\beta^T \bigr\rVert < \varepsilon/2 +\Delta_N \bigl\lVert \beta\beta^T \bigr\rVert. \]
    If $l_N<k_N$, then for $j=l_N+1, \dots, k_N$ it holds $\lVert \gamma_{t_j-}^N \rVert \geqslant \varepsilon/2$. As above, one gets
    \begin{align*}
     \bigl\lVert \gamma_u^N \bigr\rVert &\leqslant \biggl(\frac{C}{C+\varepsilon/2}\biggr)^{k_N-l_N} \bigl\lVert \Lambda_{l_N}^N \gamma_{t_{l_N}-}^N \bigr\rVert + \Delta_N \bigl\lVert \beta\beta^T \bigr\rVert \frac{2C+\varepsilon}{\varepsilon} \\
     &\leqslant \bigl\lVert \gamma_{t_{l_N}-}^N \bigr\rVert + \Delta_N \bigl\lVert \beta\beta^T \bigr\rVert \frac{2C+\varepsilon}{\varepsilon} \\
     &< \varepsilon/2 + \Delta_N \bigl\lVert \beta\beta^T \bigr\rVert \frac{2C+\varepsilon}{\varepsilon}.
    \end{align*}
    We can choose $N_1\geqslant N_0$ such that $\Delta_N \lVert \beta\beta^T \rVert \frac{2C+\varepsilon}{\varepsilon} < \varepsilon/2$ for all $N\geqslant N_1$. Then $\lVert \gamma_u^N \rVert < \varepsilon$ for all $N\geqslant N_1$.
   \end{proof}
    
   Recalling that $\gamF=\mathbf{0}_d$ for all $t\in[0,T]$, the above theorem shows that the covariance matrices $\gamma_t^{E,N}$ and $\gamma_t^{C,N}$ converge to the covariance matrix in the case of full information when the number of expert opinions on $[0,T]$ tends to infinity. As the covariance matrices contain information about the quality of the drift estimators, this means that we get an arbitrarily good estimator by increasing the number of expert opinions. In this context it does not matter whether we have an investor who observes stock returns as well as expert opinions or an investor whose only source of information are the expert opinions. Note that the assumption $\lVert\Gamma^{(N)}_k\rVert\leqslant C$ for all $N\in\N$ and $k=0, \dots, N-1$ is a way of ensuring that the experts' estimates of the drift do not get arbitrarily bad. Instead one assumes some minimal level of reliability of the experts.
   
 \section{Asymptotic Results for an Infinite Time Horizon}\label{sec:asymptotic_results_for_an_infinite_time_horizon}
  
  In the following, other than before, we consider an infinite time horizon $T=\infty$. Throughout this section we assume that the expert opinions arrive at equidistant time points $t_k=k\Delta$ for some $\Delta>0$ and with constant covariance matrix $\Gamma$.
  Our aim is to derive some results about the convergence of the conditional covariance matrices $\gamma^H_t$ for the scenario that $t$ goes to infinity.
  
  \subsection{Return Observations Only}
   
   To start with, we consider $\gamR$. The following definition that can be found in Wonham~\cite{wonham_1968} and Ku\u{c}era~\cite{kucera_1973} proves to be useful when analyzing the asymptotic behaviour of $\gamR$.
   
   \begin{definition}\label{def:stabilizable_detectable}
    We call a matrix \emph{stable} if all its eigenvalues have negative real parts. A pair $(A,B)$ of matrices $A,B\in\R^{n\times n}$ is called \emph{stabilizable} if there exists some matrix $L\in\R^{n\times n}$ such that $A+BL$ is stable. It is called \emph{detectable} if there exists some matrix $F\in\R^{n\times n}$ such that $FA+B$ is stable.
   \end{definition}
   
   We now prove that, when $t$ tends to infinity, $\gamR$ converges to some finite matrix. Here, we make use of the results from Ku\u{c}era~\cite{kucera_1973}.
   
   \begin{theorem}\label{thm:convergence_of_gamma_R}
    Consider the same model as before but with an infinite time horizon $T=\infty$. Starting with any initial covariance matrix $\Sigma_0$ it holds
    \[ \lim_{t\to\infty} \gamR = \gamma^R_{\infty} \]
    for a finite positive semidefinite matrix $\gamma^R_{\infty}$.
    Furthermore, $\gamma^R_{\infty}$ is the unique positive semidefinite solution of the algebraic Riccati equation
    \[ -\alpha\gamma-\gamma\alpha+\beta\beta^T-\gamma(\sigma\sigma^T)^{-1}\gamma = \mathbf{0}_d. \]
   \end{theorem}
   \begin{proof}
    We make use of the results in the review paper on matrix Riccati equations by Ku\u{c}era, \cite{kucera_1973}.
    After applying a simple time reversion to the differential equation considered in the paper, Theorem~17 states that the solution $P(t)$ of the differential equation
    \[ \frac{\rmd}{\rmd t} P(t) = -P(t)BB^TP(t) + P(t)A+A^TP(t) + C^TC, \qquad P(t_0) = S, \]
    satisfies
    \[ \lim_{t\to\infty} P(t) = P_{\infty} \]
    under the assumption that $(A,B)$ is stabilizable and $(C,A)$ is detectable. Theorem~5 ensures that $P_{\infty}$ is the unique positive semidefinite solution of the quadratic algebraic Riccati equation
    \[ -PBB^TP + PA+A^TP + C^TC = \mathbf{0}_d. \]
    In our model, $\gamR$ follows the dynamics
    \[ \frac{\rmd}{\rmd t} \gamR = -\alpha \gamR -\gamR\alpha + \beta\beta^T - \gamR(\sigma\sigma^T)^{-1}\gamR, \qquad \gamma_0 = \Sigma_0. \]
    Let $\tau$ denote the symmetric positive definite root of the matrix $(\sigma\sigma^T)^{-1}$, i.e.\ $\tau^2=(\sigma\sigma^T)^{-1}$. Hence, it is sufficient to show that $(-\alpha,\tau)$ is stabilizable and $(\beta^T,-\alpha)$ is detectable.
    Note that $(-\alpha)+\tau(-I_d)=-(\alpha+\tau)$ is symmetric which implies that all its eigenvalues are real. Now
    \[ \lambda_{\max}\bigl(-(\alpha+\tau)\bigr) = -\lambda_{\min}(\alpha+\tau) \leqslant -\bigl(\lambda_{\min}(\alpha)+\lambda_{\min}(\tau)\bigr) < 0, \]
    where we have used Weyl's inequality from Theorem~4.3.1 in Horn and Johnson~\cite{horn_johnson_1990} and the fact that both $\alpha$ and $\tau$ are positive definite. Hence, the pair $(-\alpha,\tau)$ is stabilizable.
    Furthermore, the matrix $(-\beta)\beta^T+(-\alpha) = -(\beta\beta^T+\alpha)$ is also symmetric and
    \[ \lambda_{\max}\bigl(-(\beta\beta^T+\alpha)\bigr) = -\lambda_{\min}(\beta\beta^T+\alpha) \leqslant -\bigl(\lambda_{\min}(\beta\beta^T)+\lambda_{\min}(\alpha)\bigr) < 0, \]
    where we have used again positive definiteness of $\alpha$ and positive semidefiniteness of $\beta\beta^T$. Hence, $(\beta^T,-\alpha)$ is detectable.
   \end{proof}
   
   In the one-dimensional situation, we get an explicit formula for $\gamma^R_{\infty}$, see Proposition~4.6 in Gabih et al.~\cite{gabih_kondakji_sass_wunderlich_2014}.
   
  \subsection{Return Observations and Expert Opinions}
   
   Now that we have seen what happens to $\gamR$ when $t$ tends to infinity, we consider the asymptotic behaviour of $\gamE$ and $\gamC$.
   
   \begin{lemma}\label{lem:monotonicity_gam_tk}
    Assume that the expert opinions arrive at equidistant time points $t_k=k\Delta$ for some $\Delta>0$, and that $\Gamma_k=\Gamma$ is some constant positive definite matrix. Let $H\in\{E,C\}$.
    If $\gamma^H_{t_0-}\leqslant \gamma^H_{t_1-}$, then $(\gamma^H_{t_k-})_{k\geqslant 0}$ and $(\gamma^H_{t_k})_{k\geqslant 0}$ are monotone non-decreasing sequences.
    If $\gamma^H_{t_0-}\geqslant \gamma^H_{t_1-}$, then $(\gamma^H_{t_k-})_{k\geqslant 0}$ and $(\gamma^H_{t_k})_{k\geqslant 0}$ are monotone non-increasing.
   \end{lemma}
   \begin{proof}
    We consider first the case $H=E$. Suppose for some $k\geqslant 1$ that $\gamma^E_{t_{k-1}-}\leqslant \gamma^E_{t_k-}$. Then clearly $\gamma^E_{t_{k-1}-}+\Gamma\leqslant \gamma^E_{t_k-}+\Gamma$ and hence $(\gamma^E_{t_{k-1}-}+\Gamma)^{-1}\geqslant (\gamma^E_{t_k-}+\Gamma)^{-1}$. It follows that
    \begin{align*}
     \gamma^E_{t_k}-\gamma^E_{t_{k-1}} &= \Gamma(\gamma^E_{t_k-}+\Gamma)^{-1}\gamma^E_{t_k-} - \Gamma(\gamma^E_{t_{k-1}-}+\Gamma)^{-1}\gamma^E_{t_{k-1}-} \\
     &= \bigl(\Gamma-\Gamma(\gamma^E_{t_k-}+\Gamma)^{-1}\Gamma\bigr) - \bigl(\Gamma-\Gamma(\gamma^E_{t_{k-1}-}+\Gamma)^{-1}\Gamma\bigr) \\
     &= \Gamma\Bigl((\gamma^E_{t_{k-1}-}+\Gamma)^{-1}-(\gamma^E_{t_k-}+\Gamma)^{-1}\Bigr)\Gamma \\
     &\geqslant \mathbf{0}_d.
    \end{align*}
    Combining this result with the formula from Lemma~\ref{lem:hatmuE_dynamics}, we see that
    \begin{align*}
     \gamma^E_{t_{k+1}-}-\gamma^E_{t_k-} &= \Bigl(\rme^{-\alpha\Delta}\gamma^E_{t_k}\rme^{-\alpha\Delta} + \int_{t_k}^{t_{k+1}} \rme^{-\alpha(t_{k+1}-s)}\beta\beta^T\rme^{-\alpha(t_{k+1}-s)}\,\rmd s\Bigr) \\
     &\qquad - \Bigl(\rme^{-\alpha\Delta}\gamma^E_{t_{k-1}}\rme^{-\alpha\Delta} + \int_{t_{k-1}}^{t_k} \rme^{-\alpha(t_k-s)}\beta\beta^T\rme^{-\alpha(t_k-s)}\,\rmd s\Bigr) \\
     &= \rme^{-\alpha\Delta}(\gamma^E_{t_k}-\gamma^E_{t_{k-1}})\rme^{-\alpha\Delta} + \int_0^{\Delta} \rme^{\alpha s}\beta\beta^T\rme^{\alpha s}\,\rmd s - \int_0^{\Delta} \rme^{\alpha s}\beta\beta^T\rme^{\alpha s}\,\rmd s \\
     &= \rme^{-\alpha\Delta}(\gamma^E_{t_k}-\gamma^E_{t_{k-1}})\rme^{-\alpha\Delta} \\
     &\geqslant \mathbf{0}_d.
    \end{align*}
    Inductively, it follows that $(\gamma^E_{t_k-})_{k\geqslant 0}$ and $(\gamma^E_{t_k})_{k\geqslant 0}$ are monotone non-decreasing. The proof that the sequences are monotone non-increasing in the case that $\gamma^E_{t_0-}\geqslant \gamma^E_{t_1-}$ goes in an analogous manner.
    
    Secondly, we consider the case $H=C$ and assume again for some $k\geqslant 1$ that $\gamma^C_{t_{k-1}-}\leqslant \gamma^C_{t_k-}$. As above, it follows from the update formula that $\gamma^C_{t_{k-1}}\leqslant \gamma^C_{t_k}$. In Lemma~\ref{lem:hatmuC_dynamics} we have seen that between two information dates $\gamC$ follows the dynamics
    \begin{equation}\label{eq:gamC_dynamics}
     \frac{\rmd}{\rmd t} \gamC = -\alpha\gamC -\gamC\alpha +\beta\beta^T - \gamC(\sigma\sigma^T)^{-1}\gamC.
    \end{equation}
    We consider the intervals $[t_{k-1},t_k)$ and $[t_k,t_{k+1})$. In both intervals, $\gamC$ evolves with the same dynamics, but for the initial values we have $\gamma^C_{t_{k-1}}\leqslant \gamma^C_{t_k}$. Since the differential equation \eqref{eq:gamC_dynamics} is a Riccati equation, it follows from Theorem~10 in Ku\u{c}era~\cite{kucera_1973} that $\gamma^C_{t_{k-1}+h}\leqslant \gamma^C_{t_k+h}$ for any time $h\in[0,\Delta)$, and in particular $\gamma^C_{t_k-}\leqslant\gamma^C_{t_{k+1}-}$. Inductively, it follows that $(\gamma^C_{t_k-})_{k\geqslant 0}$ and $(\gamma^C_{t_k})_{k\geqslant 0}$ are monotone non-decreasing sequences. The proof in the other case is again completely analogous.
   \end{proof}
   
   Under these monotonicity assumptions we can show convergence of the sequences $(\gamma^H_{t_k-})_{k\geqslant 0}$ and $(\gamma^H_{t_k})_{k\geqslant 0}$ when $k$ goes to infinity.
   
   \begin{proposition}\label{prop:limits_of_matrices}
    Let $H\in\{E,C\}$. Under the assumptions of Lemma~\ref{lem:monotonicity_gam_tk} and supposing that either $\gamma^H_{t_0-}\leqslant \gamma^H_{t_1-}$ or $\gamma^H_{t_0-}\geqslant \gamma^H_{t_1-}$, there exist finite matrices $U^H$ and $L^H$ in $\R^{d\times d}$ such that
    \[ \lim_{k\to\infty} \gamma^H_{t_k-} = U^H \quad \text{ and } \quad \lim_{k\to\infty} \gamma^H_{t_k} = L^H. \]
   \end{proposition}
   \begin{proof}
    By Lemma~\ref{lem:monotonicity_gam_tk} the sequences $(\gamma^H_{t_k-})_{k\geqslant 0}$ and $(\gamma^H_{t_k})_{k\geqslant 0}$ are monotone. Recall from Lemma~\ref{lem:hatmuE_dynamics} that between two information dates, i.e.\ for $t\in[t_k,t_{k+1})$ it holds
    \begin{align*}
     \gamE &= \rme^{-\alpha(t-t_k)}\biggl(\gamma^E_{t_k} + \int_{t_k}^t \rme^{\alpha(s-t_k)}\beta\beta^T\rme^{\alpha(s-t_k)} \,\rmd s\biggr) \rme^{-\alpha(t-t_k)} \\
     &= \rme^{-\alpha(t-t_k)}\gamma^E_{t_k}\rme^{-\alpha(t-t_k)} + \int_{t_k}^t \rme^{-\alpha(t-s)}\beta\beta^T\rme^{-\alpha(t-s)} \,\rmd s.
    \end{align*}
    Therefore for any $t\in[t_k,t_{k+1})$ we have
    \begin{align*}
     \frac{\rmd}{\rmd t} \gamE &= -\alpha\rme^{-\alpha(t-t_k)}\gamma^E_{t_k}\rme^{-\alpha(t-t_k)} - \rme^{-\alpha(t-t_k)}\gamma^E_{t_k}\rme^{-\alpha(t-t_k)}\alpha \\
     &\qquad\;\; - \alpha\rme^{-\alpha(t-t_k)}\biggl(\int_{t_k}^t \rme^{\alpha(s-t_k)}\beta\beta^T\rme^{\alpha(s-t_k)}\,\rmd s\biggr) \rme^{-\alpha(t-t_k)} + \beta\beta^T \\
     &\qquad\;\; -\rme^{-\alpha(t-t_k)}\biggl(\int_{t_k}^t \rme^{\alpha(s-t_k)}\beta\beta^T\rme^{\alpha(s-t_k)}\,\rmd s\biggr) \rme^{-\alpha(t-t_k)}\alpha \\
     &= -\alpha\gamE-\gamE\alpha+\beta\beta^T.
    \end{align*}
    This is a degenerate Riccati differential equation where the quadratic term vanishes. From Definition~\ref{def:stabilizable_detectable} it follows immediately that the pair $(-\alpha,\mathbf{0}_d)$ is stabilizable. So by Theorem~11 in Ku\u{c}era~\cite{kucera_1973} the solution of this differential equation is bounded. Since at each information date Proposition~\ref{prop:update_decreases_covariance} ensures $\gamma^E_{t_k}\leqslant\gamma^E_{t_k-}$, and by applying again Theorem~10 in \cite{kucera_1973} we can conclude that there is some matrix $M\in\R^{d\times d}$ such that $x^T\gamE x\leqslant x^TMx$ for all $x\in\R^d$ and $t\geqslant 0$. By Proposition~\ref{prop:comparison_of_gammas} the same holds for $\gamC$.
    Hence, for $H\in\{E,C\}$, the sequences $(\gamma^H_{t_k-})_{k\geqslant 0}$ and $(\gamma^H_{t_k})_{k\geqslant 0}$ are monotone and bounded. Since they are symmetric, it can be shown that
    \[ \lim_{k\to\infty} \gamma^H_{t_k-} = U^H \quad \text{ and } \quad \lim_{k\to\infty} \gamma^H_{t_k} = L^H \]
    for finite matrices $U^H$ and $L^H$ in $\R^{d\times d}$.
   \end{proof}
   
   Note that since $\gamma^H_{t_0-}=\Sigma_0$, the condition $\gamma^H_{t_0-}\leqslant\gamma^H_{t_1-}$ is trivially fulfilled in the special case $\Sigma_0=\mathbf{0}_d$, i.e.\ where the initial drift $\mu_0$ is known.
   
   For the one-dimensional situation with $d=1$ it has been shown in the proof of Proposition~4.6 in Gabih et al.~\cite{gabih_kondakji_sass_wunderlich_2014} that there exists some index $k_0\geqslant 0$ such that $\gamE$ and $\gamC$ are increasing in all intervals $[t_k, t_{k+1})$ for $k\geqslant k_0$. The question arises whether this statement can be generalized to the multidimensional situation when looking at some norm of $\gamE$ and $\gamC$. First of all, one can show that there exists some $k_0\geqslant 0$ such that the spectral norm of $\gamE$, respectively $\gamC$, is increasing in all intervals $[t_k, t_{k+1})$ for $k\geqslant k_0$ if we assume that the single stocks evolve independently. This is the case if we assume that the parameter matrices $\alpha$, $\beta$ and $\sigma\sigma^T$ as well as $\Sigma_0$ and $\Gamma$ are diagonal matrices.
   
   However, the above statement is in general not true when we have more than one stock in the market and do not assume independence of the single stocks. It is possible to find parameter sets for which the spectral norm does not become monotone between information dates. The basis for this construction is the fact that norms of solutions of Riccati differential equations in the multivariate situation are not necessarily monotone.
   
   \begin{example}\label{ex:gamma_E_not_monotone}
    We consider some specific model parameters and plot $\lVert\gamE\rVert$ for $t\in[0,5]$. The parameter matrices $\alpha$, $\beta$ and $\Sigma_0$ are chosen in such a way that the graph of the mapping $t\mapsto\lVert\tilde{\gamma}_t\rVert$ is not monotone, where the matrices $\tilde{\gamma}_t$ solve the ordinary matrix differential equation
    \[ \frac{\rmd}{\rmd t} \tilde{\gamma}_t = -\alpha\tilde{\gamma}_t - \tilde{\gamma}_t\alpha +\beta\beta^T, \qquad \tilde{\gamma}_0=\Sigma_0. \]
    This is only possible in the multivariate case since in the one-dimensional situation the solution of a Riccati differential equation is monotone.
    Now by choosing an appropriate $\Delta$ and an appropriate expert covariance matrix $\Gamma$, we can construct a situation where $\gamE$ is a periodic function. In more detail, suppose some $\Delta>0$ is chosen with $\tilde{\gamma}_{\Delta}\geqslant \tilde{\gamma}_0$, in this example $\Delta=1$. Let $L=\tilde{\gamma}_0$ and $U=\tilde{\gamma}_{\Delta}$. Now we want to find a matrix $\Gamma$ with $\Gamma(U+\Gamma)^{-1}U=L$. When assuming that both $U$ and $L$ are invertible, this comes up to setting $\Gamma=(L^{-1}-U^{-1})^{-1}$.
    
    As one can see now in Figure~\ref{fig:gamma_E_not_monotone}, by choosing the parameters stated in Table~\ref{tab:gamma_E_not_monotone} we get a situation in which $\lVert\gamE\rVert$ is periodic and not monotone between information dates. Instead, the norm of $\gamE$ drops slightly at the beginning of each interval and then increases.
    The expert's covariance matrix $\Gamma$ that has to be chosen for getting this periodic solution of $\lVert\gamE\rVert$ is approximately
    \[ \begin{pmatrix*}[r] 5.68 & 4.52 & 7.58 \\ 4.52 & 3.75 & 6.18 \\ 7.58 & 6.18 & 10.37 \end{pmatrix*}. \]
    Note that this matrix has as eigenvalues approximately 0.05, 0.11 and 19.65. This is a rather extreme covariance matrix. It corresponds to an investor who estimates some combinations of the stocks quite well but gives a rather vague estimate for some specific combination of stocks corresponding to the eigenvector of the largest eigenvalue.
   \end{example}
   
   \begin{table}[ht]
    \centering
    \setlength{\tabcolsep}{2pt}
    \begin{tabular}{p{4mm}p{3mm}p{4cm}}
     \hline\hline
     \addlinespace[2mm]
     $\alpha$ & $=$ & $\begin{pmatrix*}[r] 2.34 & -1.27 & 1.50 \\ -1.27 & 1.06 & -1.43 \\ 1.50 & -1.43 & 3.16 \end{pmatrix*}$ \\ \addlinespace[2mm]
     $\beta$ & $=$ & $\begin{pmatrix*}[r] 1.32 & -0.53 & 0.12 \\ -0.53 & 1.30 & -0.35 \\ 0.12 & -0.35 & 0.96 \end{pmatrix*}$ \\ \addlinespace[2mm]
     $\Sigma_0$ & $=$ & $\begin{pmatrix*}[r] 0.44 & -0.05 & -0.09 \\ -0.05 & 0.93 & 0.16 \\ -0.09 & 0.16 & 0.27 \end{pmatrix*}$ \\
     \addlinespace[2mm]
     \hline\hline
    \end{tabular}
    \caption{Model parameters for Example~\ref{ex:gamma_E_not_monotone}}
    \label{tab:gamma_E_not_monotone}
   \end{table}
   
   \begin{figure}[ht]
    \centering
    \setlength\figureheight{3cm}
    \setlength\figurewidth{0.7\textwidth}
    \input{gamma_E_not_monotone_new.tikz}
    \caption{Development of $\lVert\gamma^E_t\rVert$ in Example~\ref{ex:gamma_E_not_monotone}}
    \label{fig:gamma_E_not_monotone}
   \end{figure}
   
   The same construction as in Example~\ref{ex:gamma_E_not_monotone} can be made for the investor who observes stock returns as well as expert opinions.
   
   \begin{example}\label{ex:gamma_C_not_monotone}
    An example of a periodic function mapping $t$ to $\lVert\gamC\rVert$ where the norm is not monotone between information dates is given in Figure~\ref{fig:gamma_C_not_monotone}. The underlying model parameters for this example are listed in Table~\ref{tab:gamma_C_not_monotone}. In this example, we have plotted the norm of the matrices $\gamC$ over a time of three years where the time span between two information dates is assumed to be half a year. Here, the norm increases slightly at the beginning of any interval, then decreases even below its starting value and eventually increases again. In particular, in this example it does not hold that
    \[ \liminf_{t\to\infty}\; \bigl\lVert\gamC\bigr\rVert = \lim_{k\to\infty} \bigl\lVert\gamma^C_{t_k}\bigr\rVert. \]
    The resulting expert's covariance matrix $\Gamma$ is calculated in the same way as in Example~\ref{ex:gamma_E_not_monotone}. It is approximately
    \[ \begin{pmatrix*}[r] 10.30 & 8.44 & -2.66 \\ 8.44 & 7.06 & -2.30 \\ -2.66 & -2.30 & 0.82 \end{pmatrix*}. \]
    Again, we take a look at the eigenvalues of $\Gamma$. These are approximately 0.02, 0.17 and 17.99. The same phenomenon as in Example~\ref{ex:gamma_E_not_monotone} can be observed. One of the eigenvalues of $\Gamma$ is significantly larger than the others, which corresponds to an expert who estimates some combinations of the stocks very precisely and at least one rather imprecisely. The combination of stocks that the expert cannot estimate that well is given by the eigenvector to the largest eigenvalue. 
   \end{example}
   
   \begin{table}[hb]
    \centering
    \setlength{\tabcolsep}{2pt}
    \begin{tabular}{p{3mm}p{3mm}p{5cm}p{4mm}p{3mm}p{4cm}}
     \hline\hline
     \addlinespace[2mm]
     $\alpha$ & $=$ & $\begin{pmatrix*}[r] 1.36 & -2.04 & 0.75 \\ -2.04 & 3.31 & -0.81 \\ 0.75 & -0.81 & 0.94 \end{pmatrix*}$ & $\beta$ & $=$ & $\begin{pmatrix*}[r] 1.37 & 0.20 & -0.39 \\ 0.20 & 0.82 & -0.02 \\ -0.39 & -0.02 & 0.45 \end{pmatrix*}$ \\ \addlinespace[2mm]
     $\sigma$ & $=$ & $\begin{pmatrix*}[r] 0.08 & -0.13 & \phantom{-}0.15 \\ -0.07 & -0.10 & 0.13 \\ 0.12 & 0.04 & 0.04 \end{pmatrix*}$ & $\Sigma_0$ & $=$ & $\begin{pmatrix*}[r] 0.19 & 0.11 & -0.03 \\ 0.11 & 0.11 & -0.01 \\ -0.03 & -0.01 & 0.03 \end{pmatrix*}$ \\ 
     \addlinespace[2mm]
     \hline\hline
    \end{tabular}
    \caption{Model parameters for Example~\ref{ex:gamma_C_not_monotone}}
    \label{tab:gamma_C_not_monotone}
   \end{table}
   
   \begin{figure}[ht]
    \centering
    \setlength\figureheight{3cm}
    \setlength\figurewidth{0.7\textwidth}
    \input{gamma_C_not_monotone_new.tikz}
    \caption{Development of $\lVert\gamma^C_t\rVert$ in Example~\ref{ex:gamma_C_not_monotone}}
    \label{fig:gamma_C_not_monotone}
   \end{figure}
   
   The next lemma identifies one set of parameters for which it is possible to show that the norm of $\gamE$ behaves just like $\gamE$ in the one-dimensional case.
   In the following, let
   \[ G_h=\rme^{-\alpha h}\biggl(L^E+\int_0^h \rme^{\alpha s}\beta\beta^T\rme^{\alpha s}\,\rmd s\biggr) \rme^{-\alpha h} \]
   for $h\in[0, \Delta]$, in particular $G_0=L^E$ and $G_{\Delta}=U^E$. Then it holds
   \[ \frac{\rmd}{\rmd h} G_h = \rme^{-\alpha h}\bigl(-\alpha L^E-L^E\alpha+\beta\beta^T\bigr)\rme^{-\alpha h}. \]
   
   \begin{lemma}\label{lem:alpha=aI_d}
    Suppose $\alpha=a I_d$ where $a$ is some positive real number. Under the assumption that the limit matrices $U^E$ and $L^E$ exist,
    \[ -\alpha L^E-L^E\alpha+\beta\beta^T \]
    is positive semidefinite.
   \end{lemma}
   \begin{proof}
    Note that the matrix $-\alpha L^E-L^E\alpha+\beta\beta^T$ is symmetric. Suppose it had a negative eigenvalue. Let $v\in\R^d$ be a corresponding normalized eigenvector to the eigenvalue $\theta<0$, i.e.\ 
    \[ v^T\bigl(-\alpha L^E-L^E\alpha+\beta\beta^T\bigr)v = \theta <0. \]
    Define $f\colon[0,\Delta]\to\R$, $h\mapsto v^TG_hv$. Then
    \[ f'(h)=v^T\rme^{-\alpha h}\bigl(-\alpha L^E-L^E\alpha+\beta\beta^T\bigr)\rme^{-\alpha h}v. \]
    It follows that
    \[ v^TU^Ev-v^TL^Ev = v^TG_{\Delta}v-v^TG_0v = f(\Delta)-f(0) = f'(h^*) \Delta \]
    for some $h^*\in(0,\Delta)$ by the mean value theorem. Now
    \begin{align*}
     f'(h^*) &= v^T\rme^{-\alpha h^*}\bigl(-\alpha L^E-L^E\alpha+\beta\beta^T\bigr)\rme^{-\alpha h^*}v \\
     &= \rme^{-2ah^*}v^T\bigl(-\alpha L^E-L^E\alpha+\beta\beta^T\bigr)v \\
     &= \rme^{-2ah^*}\theta.
    \end{align*}
    Hence, $v^TU^Ev-v^TL^Ev=\rme^{-2ah^*}\theta\Delta<\rme^{-2a\Delta}\theta\Delta<0.$ But this is a contradiction to $L^E\leqslant U^E$. So, $-\alpha L^E-L^E\alpha+\beta\beta^T$ is positive semidefinite.
   \end{proof}
   
   We further need uniform convergence of $(\gamma^E_{t_k+h})_{k\in\N}$.
   
   \begin{lemma}\label{lem:uniform_convergence}
    Let the assumptions of Proposition~\ref{prop:limits_of_matrices} be fulfilled. Then
    \[ \lim_{k\to\infty} \max\biggl\{\sup_{h\in[0,\Delta)} \bigl\lVert\gamma^E_{t_k+h}-G_h\bigr\rVert, \; \bigl\lVert\gamma^E_{t_{k+1}-}-G_{\Delta}\bigr\rVert \biggr\} = 0. \]
   \end{lemma}
   \begin{proof}
    As in the proof of Proposition~\ref{prop:limits_of_matrices} it holds $\lim_{k\to\infty} \gamma^E_{t_k+h}= G_h$ for all $h\in[0,\Delta)$ as well as $\lim_{k\to\infty} \gamma^E_{t_{k+1}-} = G_{\Delta}$. Let $\varepsilon>0$ and $\tilde{k}\in\N$ with $\lVert\gamma^E_{t_k}-L^E\rVert<\varepsilon$ for all $k\geqslant\tilde{k}$. Then for all $h\in[0,\Delta)$ we have
    \[ \bigl\lVert\gamma^E_{t_k+h}-G_h\bigr\rVert = \bigl\lVert\rme^{-\alpha h}(\gamma^E_{t_k}-L^E)\rme^{-\alpha h}\bigr\rVert \leqslant \bigl\lVert\rme^{-\alpha h}\bigr\rVert^2\bigl\lVert\gamma^E_{t_k}-L^E\bigr\rVert < \varepsilon \]
    for all $k\geqslant\tilde{k}$. The same holds for $\lVert\gamma^E_{t_{k+1}-}-G_{\Delta}\rVert$, and the claim follows.
   \end{proof}
    
   Using the previous lemmas we can show that under the assumption $\alpha=aI_d$ the height of any decrease in $\lVert\gamE\rVert$ between $t_k$ and $t_{k+1}$ goes to zero when $k$ goes to infinity.
   
   \begin{proposition}\label{prop:decrease_in_height}
    Let the assumptions of Lemma~\ref{lem:alpha=aI_d} be fulfilled. Let $\varepsilon>0$ and fix time points $0<h_1<h_2<\Delta$. Then there exists some $k_0\in\N$ such that
    \[ \bigl\lVert\gamma^E_{t_k+h_2}\bigr\rVert \geqslant \bigl\lVert\gamma^E_{t_k+h_1}\bigr\rVert-\varepsilon \]
    for all $k\geqslant k_0$.
   \end{proposition}
   \begin{proof}
    Suppose there exist some $\varepsilon>0$ and $0<h_1<h_2<\Delta$ as well as an increasing sequence $(k_n)_{n\in\N}$ with $\lim_{n\to\infty} k_n = \infty$ such that
    \[ \bigl\lVert\gamma^E_{t_{k_n}+h_2}\bigr\rVert < \bigl\lVert\gamma^E_{t_{k_n}+h_1}\bigr\rVert-\varepsilon \]
    for all $n\in\N$.
    Define the functions $g_n\colon(0,\Delta)\to\R$, $h\mapsto\lVert\gamma^E_{t_{k_n}+h}\rVert$ for $n\in\N$. Then
    \begin{align*}
     -\varepsilon &> \bigl\lVert\gamma^E_{t_{k_n}+h_2}\bigr\rVert-\bigl\lVert\gamma^E_{t_{k_n}+h_1}\bigr\rVert = g_n(h_2)-g_n(h_1) = g_n'(h_n^*)\cdot (h_2-h_1) \\
     &= v_n^T\bigl(-\alpha\gamma^E_{t_{k_n}+h_n^*}-\gamma^E_{t_{k_n}+h_n^*}\alpha+\beta\beta^T\bigr)v_n \cdot (h_2-h_1)
    \end{align*}
    for some $h_n^*\in(h_1, h_2)$ and where $v_n$ denotes the normalized eigenvector to the largest eigenvalue of $\gamma^E_{t_{k_n}+h_n^*}$.
    But then for all $n\in\N$
    \[ v_n^T\bigl(-\alpha\gamma^E_{t_{k_n}+h_n^*}-\gamma^E_{t_{k_n}+h_n^*}\alpha+\beta\beta^T\bigr)v_n < -\frac{\varepsilon}{h_2-h_1}. \]
    This is a contradiction to the uniform convergence of $-\alpha\gamma^E_{t_k+h}-\gamma^E_{t_k+h}\alpha+\beta\beta^T$ to
    \[ -\alpha G_h-G_h\alpha+\beta\beta^T = \rme^{-\alpha h}\bigl(-\alpha L^E-L^E\alpha+\beta\beta^T\bigr)\rme^{-\alpha h} \]
    which is positive semidefinite.
   \end{proof}
   
   From the previous lemma one can conclude in particular that under the given assumptions
   \[ \liminf_{t\to\infty}\; \bigl\lVert\gamma^E_t\bigr\rVert = \bigl\lVert L^E\bigr\rVert \quad \text{and} \quad \limsup_{t\to\infty}\; \bigl\lVert\gamma^E_t\bigr\rVert = \bigl\lVert U^E\bigr\rVert. \]
   
   Similarly to the spectral norm, for both $\gamE$ and $\gamC$ there exist parameter sets for which the Frobenius norm never becomes monotone between information dates. However, when considering the Frobenius norm of the square root of $\gamE$, we can prove asymptotic bounds. For this purpose, note that the square of the Frobenius norm of $(\gamma^E_t)^{1/2}$ is the trace of $\gamma^E_t$.
   
   \begin{theorem}\label{thm:trace_of_gamma_E}
    Consider the situation with expert opinions only and suppose that the limit matrices $L^E$ and $U^E$ exist. Then we have
    \begin{align*}
     \liminf_{t\to\infty} \;\tr\bigl(\gamE\bigr) &= \lim_{k\to\infty} \tr\bigl(\gamma^E_{t_k}\bigr), \\
     \limsup_{t\to\infty} \;\tr\bigl(\gamE\bigr) &= \lim_{k\to\infty} \tr\bigl(\gamma^E_{t_k-}\bigr).
    \end{align*}
   \end{theorem}
   \begin{proof}
    We first note that for all times $t$ between information dates the trace of $\gamE$ is differentiable with
    \begin{align*}
     \ddt \tr\bigl(\gamE\bigr) &= \ddt \sum_{i=1}^d \bigl(\gamE\bigr)_{ii} = \sum_{i=1}^d \bigl(-\alpha\gamE-\gamE\alpha+\beta\beta^T\bigr)_{ii} \\
     &= \tr\bigl(-\alpha\gamE-\gamE\alpha+\beta\beta^T\bigr) = -2\tr\bigl(\alpha\gamE\bigr)+\tr\bigl(\beta\beta^T\bigr).
    \end{align*}
    Consequently, the derivative of the trace is non-negative if and only if $\tr(\alpha\gamE)\leqslant\frac{1}{2}\tr(\beta\beta^T)$.
    Now, let $(\tilde{\gamma}^E_t)_{t\geqslant0}$ be the solution of the matrix Riccati differential equation
    \[ \ddt \tilde{\gamma}^E_t = -\alpha\tilde{\gamma}^E_t-\tilde{\gamma}^E_t\alpha+\beta\beta^T, \qquad \tilde{\gamma}^E_0=\Sigma_0. \]
    Note that $\tilde{\gamma}^E$ follows the same dynamics like $\gamma^E$ but we assume for $\tilde{\gamma}^E$ that no updates take place. As in Theorem~\ref{thm:convergence_of_gamma_R} it follows from Theorem~17 in Ku\u{c}era~\cite{kucera_1973} that
    \[ \lim_{t\to\infty} \tilde{\gamma}^E_t = \tilde{\gamma}^E_{\infty} \]
    where $\tilde{\gamma}^E_{\infty}\in\R^{d\times d}$ is a symmetric positive semidefinite matrix solving
    \[ -\alpha\tilde{\gamma}^E_{\infty}-\tilde{\gamma}^E_{\infty}\alpha+\beta\beta^T=\mathbf{0}_d. \]
    Hence, we also have $\tr(-\alpha\tilde{\gamma}^E_{\infty}-\tilde{\gamma}^E_{\infty}\alpha+\beta\beta^T)=0$, i.e.\ $\tr(\alpha\tilde{\gamma}^E_{\infty})=\frac{1}{2}\tr(\beta\beta^T)$.
    In the following we prove an asymptotic bound for the minimal eigenvalue of $\tilde{\gamma}^E_{\infty}-\gamma^E_{t_k+h}$ where $h$ ranges from 0 to $\Delta$ and $k$ goes to infinity.
    To avoid cumbersome notation we shall write $\gamma^E_{t_k+\Delta}$ when actually meaning $\gamma^E_{t_{k+1}-}$, i.e.\ the limit of the covariance matrix before the update takes place.
    Using Weyl's inequality, stated for example in Theorem~4.3.1 in Horn and Johnson~\cite{horn_johnson_1990}, we get
    \begin{align*}
     \min_{h\in[0,\Delta]} \lambda_{\min}\bigl(\tilde{\gamma}^E_{\infty}-\gamma^E_{t_k+h}\bigr) &= \min_{h\in[0,\Delta]} \lambda_{\min}\bigl(\tilde{\gamma}^E_{\infty}-G_h+G_h-\gamma^E_{t_k+h}\bigr) \\
     &\geqslant \min_{h\in[0,\Delta]} \lambda_{\min}\bigl(\tilde{\gamma}^E_{\infty}-G_h\bigr) + \lambda_{\min}\bigl(G_h-\gamma^E_{t_k+h}\bigr) \\
     &\geqslant \min_{h\in[0,\Delta]} \lambda_{\min}\bigl(\tilde{\gamma}^E_{\infty}-G_h\bigr) + \min_{h\in[0,\Delta]} \lambda_{\min}\bigl(G_h-\gamma^E_{t_k+h}\bigr).
    \end{align*}
    Now we use that $\tilde{\gamma}^E_{\infty}=\lim_{t\to\infty} \tilde{\gamma}^E_t = \lim_{k\to\infty} \tilde{\gamma}^E_{t_k+h}$ and $G_h=\lim_{k\to\infty} \gamma^E_{t_k+h}$ for any $h\in[0,\Delta]$ as well as the fact that $\gamE\leqslant\tilde{\gamma}^E_t$ for all $t\geqslant0$. This follows from Theorem~10 in Ku\u{c}era~\cite{kucera_1973} together with the proofs of Lemma~\ref{lem:monotonicity_gam_tk} and Proposition~\ref{prop:limits_of_matrices}. Combining these results we get
    \[ x^T\bigl(\tilde{\gamma}^E_{\infty}-G_h\bigr)x = \lim_{k\to\infty} x^T\bigl(\tilde{\gamma}^E_{t_k+h}-\gamma^E_{t_k+h}\bigr)x\geqslant0 \]
    for all $x\in\R^d$ and $h\in[0,\Delta]$. Consequently,
    \[ \min_{h\in[0,\Delta]} \lambda_{\min}\bigl(\tilde{\gamma}^E_{\infty}-G_h\bigr)\geqslant0. \]
    In Lemma~\ref{lem:uniform_convergence} we have shown uniform convergence of $(\gamma^E_{t_k+h})_{k\in\N}$. This implies that
    \[ \lim_{k\to\infty} \min_{h\in[0,\Delta]} \lambda_{\min}\bigl(G_h-\gamma^E_{t_k+h}\bigr) =0. \]
    Together with the above inequality this yields
    \[ \lim_{k\to\infty} \min_{h\in[0,\Delta]} \lambda_{\min}\bigl(\tilde{\gamma}^E_{\infty}-\gamma^E_{t_k+h}\bigr)\geqslant0. \]
    By applying a trace inequality proven in Lemma~1 in Wang, Kuo and Hsu~\cite{wang_kuo_hsu_1986} we can now conclude
    \[ \lim_{k\to\infty} \min_{h\in[0,\Delta]} \tr\Bigl(\alpha\bigl(\tilde{\gamma}^E_{\infty}-\gamma^E_{t_k+h}\bigr)\Bigr) \geqslant \lim_{k\to\infty} \min_{h\in[0,\Delta]} \tr(\alpha)\lambda_{\min}\bigl(\tilde{\gamma}^E_{\infty}-\gamma^E_{t_k+h}\bigr) \geqslant0, \]
    and therefore
    \begin{equation}\label{eq:limit_of_trace_alpha_gamma}
     \lim_{k\to\infty} \max_{h\in[0,\Delta]} \tr\bigl(\alpha\gamma^E_{t_k+h}\bigr) \leqslant \tr\bigl(\alpha\tilde{\gamma}^E_{\infty}\bigr) = \frac{1}{2}\tr\bigl(\beta\beta^T\bigr).
    \end{equation}
    The remaining part of the proof goes just like the proof of Proposition~\ref{prop:decrease_in_height}. Suppose there exist some $\varepsilon>0$ and $0<h_1<h_2<\Delta$ as well as an increasing sequence $(k_n)_{n\in\N}$ with $\lim_{n\to\infty} k_n=\infty$ such that
    \[ \tr\bigl(\gamma^E_{t_{k_n}+h_2}\bigr)<\tr\bigl(\gamma^E_{t_{k_n}+h_1}\bigr)-\varepsilon \]
    for all $n\in\N$. Define the functions $g_n\colon(0,\Delta)\to\R$, $h\mapsto\tr(\gamma^E_{t_{k_n}+h})$ for $n\in\N$. For all $n\in\N$ we then have
    \begin{align*}
     -\varepsilon &> \tr\bigl(\gamma^E_{t_{k_n}+h_2}\bigr)-\tr\bigl(\gamma^E_{t_{k_n}+h_1}\bigr) = g_n(h_2)-g_n(h_1) = g_n'(h_n^*)\cdot(h_2-h_1) \\
     &= \Bigl(-2\tr\bigl(\alpha\gamma^E_{t_{k_n}+h_n^*}\bigr)+\tr\bigl(\beta\beta^T\bigr)\Bigr)\cdot(h_2-h_1),
    \end{align*}
    for some $h_n^*\in(h_1,h_2)$. Hence
    \[ -2\tr\bigl(\alpha\gamma^E_{t_{k_n}+h_n^*}\bigr)+\tr\bigl(\beta\beta^T\bigr)<-\frac{\varepsilon}{h_2-h_1} \]
    for all $n\in\N$. But this is a contradiction to \eqref{eq:limit_of_trace_alpha_gamma}. Hence the assumption was wrong and we can conclude in particular that
    \[ \liminf_{t\to\infty} \;\tr\bigl(\gamE\bigr) = \lim_{k\to\infty} \tr\bigl(\gamma^E_{t_k}\bigr) \quad \text{and} \quad \limsup_{t\to\infty} \;\tr\bigl(\gamE\bigr) = \lim_{k\to\infty} \tr\bigl(\gamma^E_{t_k-}\bigr). \]
   \end{proof}
   
   For an investor who observes stock returns as well as expert opinions the above statement does not hold in general.
   
   \begin{example}\label{ex:trace_counterexample}
    Figure~\ref{fig:trace_counterexample} shows the trace of $\gam{C}{t}$ plotted over time for some exemplary parameters. In this example we have a financial market with $d=3$ stocks and one expert opinion each year, i.e.\ $\Delta=1$. The remaining model parameters are listed in Table~\ref{tab:trace_counterexample}. Note that the chosen expert matrix $\Gamma$ is approximately
    \[ \begin{pmatrix*}[r] 155.14 & 13.25 & -36.59 \\ 13.25 & 1.40 & -3.13 \\ -36.59 & -3.13 & 8.83 \end{pmatrix*} \]
    with eigenvalues approximately $164.92$, $0.19$ and $0.26$.
    The trace of $\gam{C}{t}$ decreases slightly right after any information date before eventually increasing until the next expert opinion arrives. Since we have, as in Example~\ref{ex:gamma_E_not_monotone} and Example~\ref{ex:gamma_C_not_monotone}, constructed the corresponding $\Gamma$ in such a way that $\tr(\gam{C}{t})$ is a periodic function, this shows that the claim from the previous theorem does not hold when replacing the $E$-investor with the $C$-investor.
   \end{example}
   
   \begin{table}[hb]
    \centering
    \setlength{\tabcolsep}{2pt}
    \begin{tabular}{p{3mm}p{3mm}p{5cm}p{4mm}p{3mm}p{4cm}}
     \hline\hline
     \addlinespace[2mm]
     $\alpha$ & $=$ & $\begin{pmatrix*}[r] 2.38 & 1.08 & -1.47 \\ 1.08 & 1.19 & -0.75 \\ -1.47 & -0.75 & 2.74 \end{pmatrix*}$ & $\beta$ & $=$ & $\begin{pmatrix*}[r] -7.68 & 9.76 & -3.26 \\ -4.38 & -9.01 & -0.52 \\ -0.30 & -4.51 & 2.98 \end{pmatrix*}$ \\ \addlinespace[2mm]
     $\sigma$ & $=$ & $\begin{pmatrix*}[r] -0.95 & -0.56 & -0.37 \\ -0.80 & -0.27 & 0.93 \\ 0.39 & -0.19 & 0.97 \end{pmatrix*}$ & $\Sigma_0$ & $=$ & $\begin{pmatrix*}[r] \phantom{-}6.07 & 2.63 & 0.36 \\ 2.63 & 2.38 & -0.02 \\ 0.36 & -0.02 & 0.94 \end{pmatrix*}$ \\ 
     \addlinespace[2mm]
     \hline\hline
    \end{tabular}
    \caption{Model parameters for Example~\ref{ex:trace_counterexample}}
    \label{tab:trace_counterexample}
   \end{table}
   
   \begin{figure}[ht]
    \centering
    \setlength\figureheight{5cm}
    \setlength\figurewidth{0.8\textwidth}
    \input{trace_counterexample.tikz}
    \caption{Development of $\tr(\gamma^C_t)$ in Example~\ref{ex:trace_counterexample}}
    \label{fig:trace_counterexample}
   \end{figure}
   
   Next, we identify one condition on the parameters such that the statement also holds for $H=C$.
   
   \begin{proposition}
    Assume that $\sigma\sigma^T=sI_d$ for some $s>0$ and that $(\gam{C}{t_k+h})_{k\in\N}$ converges uniformly. Then
    \begin{align*}
     \liminf_{t\to\infty} \;\tr\bigl(\gam{C}{t}\bigr) &= \lim_{k\to\infty} \tr\bigl(\gam{C}{t_k}\bigr), \\
     \limsup_{t\to\infty} \;\tr\bigl(\gam{C}{t}\bigr) &= \lim_{k\to\infty} \tr\bigl(\gam{C}{t_k-}\bigr).
    \end{align*}
   \end{proposition}
   \begin{proof}
    First, we note that between information dates the trace of $\gam{C}{t}$ is differentiable with
    \begin{align*}
     \ddt \tr\bigl(\gam{C}{t}\bigr) &= \tr\bigl(-\alpha\gam{C}{t}-\gam{C}{t}\alpha+\beta\beta^T-\gam{C}{t}(\sigma\sigma^T)^{-1}\gam{C}{t}\bigr) \\
     &= -2\tr\bigl(\alpha\gam{C}{t}\bigr)+\tr\bigl(\beta\beta^T\bigr)-\tr\bigl(\gam{C}{t}(\sigma\sigma^T)^{-1}\gam{C}{t}\bigr).
    \end{align*}
    Hence it holds $\ddt \tr(\gam{C}{t})\geqslant0$ if and only if
    \[ 2\tr\bigl(\alpha\gam{C}{t}\bigr)+\tr\bigl(\gam{C}{t}(\sigma\sigma^T)^{-1}\gam{C}{t}\bigr)\leqslant\tr\bigl(\beta\beta^T\bigr). \]
    Furthermore, we have shown that $\lim_{t\to\infty} \gam{R}{t}=\gam{R}{\infty}$ where
    \[ -\alpha\gam{R}{\infty}-\gam{R}{\infty}\alpha+\beta\beta^T-\gam{R}{\infty}(\sigma\sigma^T)^{-1}\gam{R}{\infty}=\mathbf{0}_d. \]
    It follows that
    \[ \tr\bigl(-\alpha\gam{R}{\infty}-\gam{R}{\infty}\alpha+\beta\beta^T-\gam{R}{\infty}(\sigma\sigma^T)^{-1}\gam{R}{\infty}\bigr)=0, \]
    i.e.\
    \[ 2\tr\bigl(\alpha\gam{R}{\infty}\bigr)+\tr\bigl(\gam{R}{\infty}(\sigma\sigma^T)^{-1}\gam{R}{\infty}\bigr)=\tr\bigl(\beta\beta^T\bigr). \]
    Again using Weyl's inequality we deduce
    \begin{align*}
     \min_{h\in[0,\Delta]} \lambda_{\min}\bigl(\gam{R}{\infty}-\gam{C}{t_k+h}\bigr) &= \min_{h\in[0,\Delta]} \lambda_{\min}\bigl(\gam{R}{\infty}-G^C_h+G^C_h-\gam{C}{t_k+h}\bigr) \\
     &\geqslant \min_{h\in[0,\Delta]} \lambda_{\min}\bigl(\gam{R}{\infty}-G^C_h\bigr)+\min_{h\in[0,\Delta]} \lambda_{\min}\bigl(G^C_h-\gam{C}{t_k+h}\bigr),
    \end{align*}
    where $G^C_h=\lim_{k\to\infty} \gam{C}{t_k+h}$ for any $h\in[0,\Delta]$. Recall that $\gam{R}{\infty}=\lim_{t\to\infty}\gam{R}{t}=\lim_{k\to\infty}\gam{R}{t_k+h}$ for each $h\in[0,\Delta]$ and that $\gam{C}{t}\leqslant\gam{R}{t}$ for each $t\geqslant0$, see Proposition~\ref{prop:comparison_of_gammas}. This implies that
    \[ x^T\bigl(\gam{R}{\infty}-G^C_h\bigr)x = \lim_{k\to\infty} x^T\bigl(\gam{R}{t_k+h}-\gam{C}{t_k+h}\bigr)x \geqslant0 \]
    for all $x\in\R^d$, $h\in[0,\Delta]$, and therefore
    \[ \min_{h\in[0,\Delta]}\lambda_{\min}\bigl(\gam{R}{\infty}-G^C_h\bigr)\geqslant0. \]
    By uniform convergence of $(\gam{C}{t_k+h})_{k\in\N}$ to $G^C_h$ we get
    \[ \lim_{k\to\infty} \min_{h\in[0,\Delta]}\lambda_{\min}\bigl(G^C_h-\gam{C}{t_k+h}\bigr)=0. \]
    Hence, putting these results together, we obtain
    \[ \lim_{k\to\infty} \min_{h\in[0,\Delta]} \lambda_{\min}\bigl(\gam{R}{\infty}-\gam{C}{t_k+h}\bigr)\geqslant0. \]
    In order to prove our claim we need to show that
    \[ \lim_{k\to\infty} \max_{h\in[0,\Delta]} 2\tr\bigl(\alpha\gam{C}{t_k+h}\bigr)+\tr\bigl(\gam{C}{t_k+h}(\sigma\sigma^T)^{-1}\gam{C}{t_k+h}\bigr)\leqslant\tr\bigl(\beta\beta^T\bigr), \]
    where the right-hand side is equal to $2\tr(\alpha\gam{R}{\infty})+\tr(\gam{R}{\infty}(\sigma\sigma^T)^{-1}\gam{R}{\infty})$. Putting this together and using cyclicity of the trace we see that we need to prove
    \begin{equation}\label{eq:claimed_inequality}
     \lim_{k\to\infty} \min_{h\in[0,\Delta]} 2\tr\Bigl(\alpha\bigl(\gam{R}{\infty}-\gam{C}{t_k+h}\bigr)\Bigr)+\tr\Bigl((\sigma\sigma^T)^{-1}\bigl((\gam{R}{\infty})^2-(\gam{C}{t_k+h})^2\bigr)\Bigr)\geqslant0.
    \end{equation}
    By Lemma~1 in Wang, Kuo and Hsu~\cite{wang_kuo_hsu_1986} it follows
    \[ \lim_{k\to\infty} \min_{h\in[0,\Delta]} 2\tr\Bigl(\alpha\bigl(\gam{R}{\infty}-\gam{C}{t_k+h}\bigr)\Bigr) \geqslant \lim_{k\to\infty} \min_{h\in[0,\Delta]} 2\tr(\alpha)\lambda_{\min}\bigl(\gam{R}{\infty}-\gam{C}{t_k+h}\bigr) \geqslant0. \]
    For the second summand we make use of our assumption $\sigma\sigma^T=sI_d$ where $s>0$. It follows that
    \[ \tr\Bigl((\sigma\sigma^T)^{-1}\bigl((\gam{R}{\infty})^2-(\gam{C}{t_k+h})^2\bigr)\Bigr)=\frac{1}{s}\cdot\tr\bigl((\gam{R}{\infty})^2-(\gam{C}{t_k+h})^2\bigr). \]
    Similarly to above, we write
    \[ \lim_{k\to\infty} \min_{h\in[0,\Delta]} \tr\bigl((\gam{R}{\infty})^2-(\gam{C}{t_k+h})^2\bigr) = \min_{h\in[0,\Delta]} \tr\bigl((\gam{R}{\infty})^2-(G^C_h)^2\bigr)+\lim_{k\to\infty} \min_{h\in[0,\Delta]} \tr\bigl((G^C_h)^2-(\gam{C}{t_k+h})^2\bigr). \]
    By uniform convergence, the second summand is zero. For the first summand, we recall that for any $h\in[0,\Delta]$ we have shown $G^C_h\leqslant\gam{R}{\infty}$ where both matrices are symmetric positive semidefinite. It follows that
    \[ \tr\bigl((G^C_h)^2\bigr)\leqslant\tr\bigl((\gam{R}{\infty})^2\bigr). \]
    Putting these results together, we obtain
    \[ \lim_{k\to\infty} \min_{h\in[0,\Delta]} \tr\bigl((\gam{R}{\infty})^2-(\gam{C}{t_k+h})^2\bigr) \geqslant0. \]
    Now, the inequality in \eqref{eq:claimed_inequality} has been shown.
   \end{proof}
   
   Instead of requiring $\sigma\sigma^T$ to be a multiple of the unit matrix, we can also put some restriction on the form of the expert's covariance matices $\Gamma_k$ to ensure monotonicity of $\gamC$ in between information dates. Here, the information dates $t_k$ are again arbitrary and we allow for non-constant $\Gamma_k$.
   
   \begin{proposition}
    Suppose that the initial covariance matrix $\Sigma_0$ is positive definite and fulfills
    \[ -\alpha\Sigma_0-\Sigma_0\alpha+\beta\beta^T-\Sigma_0(\sigma\sigma^T)^{-1}\Sigma_0 \geqslant\mathbf{0}_d \]
    and that the expert's covariance matrices are of the form $\Gamma_k=c_k\gam{C}{t_k-}$ for some $c_k>0$ at any information date $t_k$.
    Then between any two successive information dates, $\gam{C}{t}$ is non-decreasing in the sense of the positive semidefinite ordering.
   \end{proposition}
   \begin{proof}
    Suppose that
    \[ -\alpha\gam{C}{t_k-}-\gam{C}{t_k-}\alpha+\beta\beta^T-\gam{C}{t_k-}(\sigma\sigma^T)^{-1}\gam{C}{t_k-} \geqslant\mathbf{0}_d \]
    for some $k\geqslant0$. We look at the covariance matrix after the update.
    First, note that
    \[ \gam{C}{t_k}=\Gamma_k\bigl(\gam{C}{t_k-}+\Gamma_k\bigr)^{-1}\gam{C}{t_k-}=c_k\gam{C}{t_k-}\bigl(\gam{C}{t_k-}+c_k\gam{C}{t_k-}\bigr)^{-1}\gam{C}{t_k-}=\frac{c_k}{1+c_k}\gam{C}{t_k-}. \]
    We write $d_k=\frac{c_k}{1+c_k}$ and note that $d_k<1$. Hence, the updated covariance matrix is just a multiple of the matrix before the update. Now we can write
    \begin{align*}
     &-\alpha\gam{C}{t_k}-\gam{C}{t_k}\alpha+\beta\beta^T-\gam{C}{t_k}(\sigma\sigma^T)^{-1}\gam{C}{t_k}\\
     &\qquad = d_k\bigl(-\alpha\gam{C}{t_k-}-\gam{C}{t_k-}\alpha\bigr)+\beta\beta^T-d_k^2\gam{C}{t_k-}(\sigma\sigma^T)^{-1}\gam{C}{t_k-} \\
     &\qquad = d_k\bigl(-\alpha\gam{C}{t_k-}-\gam{C}{t_k-}\alpha+\beta\beta^T-\gam{C}{t_k-}(\sigma\sigma^T)^{-1}\gam{C}{t_k-}\bigr)+(1-d_k)\beta\beta^T+(d_k-d_k^2)\gam{C}{t_k-}(\sigma\sigma^T)^{-1}\gam{C}{t_k-}.
    \end{align*}
    Since $d_k-d_k^2>0$, and by assumption on $\gam{C}{t_k-}$, this is a sum of positive semidefinite matrices, and hence itself a positive semidefinite matrix.
    By Theorem~2.1 in Rodriguez-Canabal~\cite{rodriguez-canabal_1975} it follows that $\gam{C}{t}$ is monotone non-decreasing in the interval $[t_k,t_{k+1})$ and that
    \[ -\alpha\gam{C}{t_{k+1}-}-\gam{C}{t_{k+1}-}\alpha+\beta\beta^T-\gam{C}{t_{k+1}-}(\sigma\sigma^T)^{-1}\gam{C}{t_{k+1}-} \geqslant\mathbf{0}_d. \]
    Inductively, it follows that $\gam{C}{t}$ is monotone non-decreasing between any two successive information dates.
   \end{proof}
   
   The above theorem implies in particular that the trace of $\gam{C}{t}$ as well as its spectral norm is increasing between any two successive information dates. Hence, under the assumptions of the theorem we also deduce
   \begin{align*}
    \liminf_{t\to\infty} \;\tr\bigl(\gam{C}{t}\bigr) = \lim_{k\to\infty} \tr\bigl(\gam{C}{t_k}\bigr) \quad &\text{and} \quad \liminf_{t\to\infty} \;\bigl\lVert\gam{C}{t}\bigr\rVert = \lim_{k\to\infty} \bigl\lVert\gam{C}{t_k}\bigr\rVert, \\
    \limsup_{t\to\infty} \;\tr\bigl(\gam{C}{t}\bigr) = \lim_{k\to\infty} \tr\bigl(\gam{C}{t_k-}\bigr) \quad &\text{and} \quad \limsup_{t\to\infty} \;\bigl\lVert\gam{C}{t}\bigr\rVert = \lim_{k\to\infty} \bigl\lVert\gam{C}{t_k-}\bigr\rVert.
   \end{align*}
   
 \section{Portfolio Optimization Problem}\label{sec:portfolio_optimization_problem}
  
  \subsection{Optimal Strategy and Value Function}
   
   An investor's trading is described by a self-financing trading strategy $\pi=(\pi_t)_{t\in[0,T]}$ where $\pi_t$ takes values in $\R^d$. Here, the value $\pi^i_t$, $i=1, \dots, d$, represents the proportion of wealth invested in stock $i$ at time $t$, while the proportion $1-\ones^T\pi_t$ is invested in the risk-free bond $S^0$. Here $\mathbf{1}_d\in\R^d$ denotes the vector consisting of ones. Let $X^\pi=(X^\pi_t)_{t\in[0,T]}$ denote the wealth process corresponding to $\pi$. For the dynamics of the wealth process we get
   \[ \rmd X^\pi_t = X^\pi_t \Bigl(\pi_t^T\bigl((\mu_t-r_t\mathbf{1}_d)\,\rmd t + \sigma \,\rmd W_t\bigr) +r_t\,\rmd t\Bigr). \]
   The investor's initial capital at time zero is $X^\pi_0=x_0>0$.
   We denote by
   \[ \mathcal{A}^H(x_0) = \biggl\{\pi=(\pi_t)_{t\in[0,T]} \;\bigg|\; \pi \text{ is } \mathbb{F}^H\text{-adapted}, \; X^\pi_0=x_0, \; \E\biggl[\int_0^T \lVert\sigma^T\pi_t\rVert^2\,\rmd t\biggr]<\infty\biggr\} \]
   the class of admissible trading strategies where $H\in\{R,E,C,F\}$. Here, and in everything that follows, when $v\in\R^d$ is some vector then $\lVert v\rVert$ denotes the Euclidean norm of $v$. The objective of our portfolio optimization problem is to maximize expected logarithmic utility of terminal wealth. We call
   \[ V^H(x_0) = \sup\Bigl\{\E[\log(X^\pi_T)] \;\Big|\; \pi\in\mathcal{A}^H(x_0)\Bigr\} \]
   the \emph{value function} of the optimization problem.
   One can show that for $H\in\{R,E,C,F\}$ and writing $\hat{\mu}_t$ for $\hat{\mu}^H_t$ it holds
   \begin{align*}
    & \E\bigl[(\hat{\mu}_t-r_t\ones)^T(\sigma\sigma^T)^{-1}(\hat{\mu}_t-r_t\ones)\bigr] \\
    & \qquad = \tr\bigl((\sigma\sigma^T)^{-1}\E[\hat{\mu}_t(\hat{\mu}_t)^T]\bigr)-2(r_t\ones)^T(\sigma\sigma^T)^{-1}m_t +(r_t\ones)^T(\sigma\sigma^T)^{-1}(r_t\ones),
   \end{align*}
   where $m_t$ is the mean of $\mu_t$.
   
   Using this equality it is possible to calculate the optimal strategy for our optimization problem and to show that it is admissible.
   
   \begin{proposition}\label{prop:optimal_strategy}
    Let $H\in\{R,E,C,F\}$. The optimal strategy for the optimization problem
    \[ V^H(x_0) = \sup\Bigl\{\E[\log(X^\pi_T)] \;\Big|\; \pi\in\mathcal{A}^H(x_0)\Bigr\} \]
    is $\pi^*=(\pi^*_t)_{t\in[0,T]}$ with $\pi^*_t=(\sigma\sigma^T)^{-1}(\hat{\mu}^H_t-r_t\ones)$.
   \end{proposition}
   \begin{proof}
    From the dynamics of the wealth process we get for any $\pi\in\mathcal{A}^H(x_0)$ that
    \[ \log(X^\pi_T) = \log(x_0)+\int_0^T \bigl(\pi_t^T(\mu_t-r_t\mathbf{1}_d)+r_t-\frac{1}{2}\lVert \sigma^T\pi_t\rVert^2 \bigr) \,\rmd t + \int_0^T \pi_t^T\sigma \,\rmd W_t.\]
    We now apply Fubini and use that, since $\pi\in\mathcal{A}^H(x_0)$, the stochastic integral has expectation zero. Hence, we deduce
    \begin{align*}
     \E[\log(X^\pi_T)] &= \log(x_0) + \int_0^T \E\bigl[\pi_t^T(\mu_t-r_t\mathbf{1}_d)+r_t-\frac{1}{2}\lVert \sigma^T\pi_t\rVert^2\bigr] \,\rmd t \\
     &= \log(x_0) + \int_0^T \E\Bigl[\E\bigl[\pi_t^T(\mu_t-r_t\mathbf{1}_d)+r_t-\frac{1}{2}\lVert \sigma^T\pi_t\rVert^2\;\big|\; \mathcal{F}^H_t\bigr]\Bigr] \,\rmd t \\
     &= \log(x_0) + \int_0^T \E\bigl[\pi_t^T(\hat{\mu}^H_t-r_t\mathbf{1}_d)+r_t-\frac{1}{2}\lVert \sigma^T\pi_t\rVert^2\bigr] \,\rmd t.
    \end{align*}
    Now we fix some $t\in[0,T]$. Following a pointwise maximization, we formally take the derivative of the expression inside the expectation with respect to $\pi_t$. Using the first-order condition, we set the derivative to zero, which means setting $\hat{\mu}^H_t-r_t\mathbf{1}_d-\sigma\sigma^T\pi_t$ equal to the zero vector. Since we have assumed that $\sigma\sigma^T$ is positive definite, this implies that $\pi^*_t=(\sigma\sigma^T)^{-1}(\hat{\mu}^H_t-r_t\ones)$ maximizes the above integrand pointwise.
    It remains to check that $(\pi^*_t)_{t\in[0,T]}$ is indeed admissible. First, we note that
    \begin{align*}
     \int_0^T \lVert \sigma^T\pi^*_t \rVert^2 \,\rmd t &= \int_0^T \bigl\lVert \sigma^T(\sigma\sigma^T)^{-1}(\hat{\mu}^H_t-r_t\ones) \bigr\rVert^2 \,\rmd t \\
     &= \int_0^T (\hat{\mu}^H_t-r_t\ones)^T(\sigma\sigma^T)^{-1}(\hat{\mu}^H_t-r_t\ones) \,\rmd t.
    \end{align*}
    Taking the expectation and applying Fubini we get
    \begin{align*}
     &\E\biggl[\int_0^T \lVert \sigma^T\pi^*_t \rVert^2 \,\rmd t\biggr] = \int_0^T \E\bigl[(\hat{\mu}^H_t-r_t\ones)^T(\sigma\sigma^T)^{-1}(\hat{\mu}^H_t-r_t\ones)\bigr] \,\rmd t \\
     &\qquad\;\; = \int_0^T \tr\bigl((\sigma\sigma^T)^{-1}\E[\hat{\mu}_t(\hat{\mu}_t)^T]\bigr)-2(r_t\ones)^T(\sigma\sigma^T)^{-1}m_t + (r_t\ones)^T(\sigma\sigma^T)^{-1}(r_t\ones) \,\rmd t,
    \end{align*}
    where the last equality has been stated above. The integrals over the last two summands are finite due to continuity of both $(r_t)_{t\in[0,T]}$ and $(m_t)_{t\in[0,T]}$. We consider the first summand in more detail. We have $\E[\hat{\mu}_t(\hat{\mu}_t)^T]= \Sigma_t+m_tm_t^T-\gamma^H_t$ by Equation~\eqref{eq:second_moments_filter}. Hence,
    \begin{align*}
     \int_0^T \tr\bigl((\sigma\sigma^T)^{-1}\E[\hat{\mu}_t(\hat{\mu}_t)^T]\bigr) \,\rmd t &= \int_0^T \tr\bigl((\sigma\sigma^T)^{-1}(\Sigma_t+m_tm_t^T-\gamma^H_t)\bigr)\,\rmd t \\
     &= \int_0^T \tr\bigl((\sigma\sigma^T)^{-1}(\Sigma_t+m_tm_t^T)\bigr) - \tr\bigl((\sigma\sigma^T)^{-1}\gamma^H_t\bigr) \,\rmd t.
    \end{align*}
    Recall that $(\sigma\sigma^T)^{-1}$ is symmetric positive definite and $\gamma^H_t$ is symmetric positive semidefinite. It can be shown that the product $(\sigma\sigma^T)^{-1}\gamma^H_t$ has a non-negative trace, hence
    \[ \int_0^T \tr\bigl((\sigma\sigma^T)^{-1}\E[\hat{\mu}_t(\hat{\mu}_t)^T]\bigr) \,\rmd t \leqslant \int_0^T \tr\bigl((\sigma\sigma^T)^{-1}(\Sigma_t+m_tm_t^T)\bigr) \,\rmd t, \]
    which is finite due to continuity. It follows that
    \[ \E\biggl[\int_0^T \lVert \sigma^T\pi^*_t \rVert^2 \,\rmd t\biggr] < \infty, \]
    so $(\pi^*_t)_{t\in[0,T]}$ is an admissible strategy.
   \end{proof}
   
   Note that under full information the optimal strategy is $(\sigma\sigma^T)^{-1}(\mu_t-r_t\ones)$. That means that for our portfolio optimization problem under partial information, the \textit{certainty equivalence principle} holds, meaning that the drift $\mu_t$ in the optimal strategy is replaced by the filter $\hat{\mu}^H_t$. Now that we have an explicit formula for the optimal trading strategy it is easy to write down the optimal value function for the optimization problem.
   
   \begin{theorem}\label{thm:optimal_value}
    The optimal value of the portfolio optimization problem is
    \begin{align*}
     V^H(x_0) &= \log(x_0)+\int_0^T \frac{1}{2}\E\bigl[(\hat{\mu}^H_t-r_t\ones)^T(\sigma\sigma^T)^{-1}(\hat{\mu}^H_t-r_t\ones)\bigr]+r_t \,\rmd t \\
     &= \log(x_0)+\int_0^T r_t-(r_t\ones)^T(\sigma\sigma^T)^{-1}m_t+\frac{1}{2}(r_t\ones)^T(\sigma\sigma^T)^{-1}(r_t\ones)\,\rmd t \\
     & \qquad\;\; +\frac{1}{2}\int_0^T \tr\bigl((\sigma\sigma^T)^{-1}(\Sigma_t+m_tm_t^T-\gamma^H_t)\bigr)\,\rmd t.
    \end{align*}
   \end{theorem}
   \begin{proof}
    Throughout the proof we shortly write $\hat{\mu}_t$ for $\hat{\mu}^H_t$. As above,
    \[ V^H(x_0) = \log(x_0)+\int_0^T \E\bigl[(\pi^*_t)^T(\hat{\mu}_t-r_t\mathbf{1}_d)+r_t-\frac{1}{2}\lVert \sigma^T\pi^*_t\rVert^2\bigr] \,\rmd t. \]
    After inserting the optimal strategy from Proposition~\ref{prop:optimal_strategy} the integral term becomes
    \[ \int_0^T \E\Bigl[(\hat{\mu}_t-r_t\mathbf{1}_d)^T(\sigma\sigma^T)^{-1}(\hat{\mu}_t-r_t\mathbf{1}_d)+r_t-\frac{1}{2}\bigl\lVert (\hat{\mu}_t-r_t\mathbf{1}_d)^T(\sigma\sigma^T)^{-1}\sigma\bigr\rVert^2\Bigr]\,\rmd t. \]
    The last summand inside the expectation can be written as
    \begin{align*}
     \frac{1}{2}\bigl\lVert (\hat{\mu}_t-r_t\mathbf{1}_d)^T(\sigma\sigma^T)^{-1}\sigma\bigr\rVert^2 &= \frac{1}{2}\Bigl((\hat{\mu}_t-r_t\mathbf{1}_d)^T(\sigma\sigma^T)^{-1}\sigma\Bigr) \Bigl((\hat{\mu}_t-r_t\mathbf{1}_d)^T(\sigma\sigma^T)^{-1}\sigma\Bigr)^T \\
     &= \frac{1}{2}(\hat{\mu}_t-r_t\mathbf{1}_d)^T(\sigma\sigma^T)^{-1}(\hat{\mu}_t-r_t\mathbf{1}_d),
    \end{align*}
    so all in all we get for the value function
    \[ V^H(x_0) = \log(x_0) +\int_0^T \E\Bigl[\frac{1}{2}(\hat{\mu}_t-r_t\mathbf{1}_d)^T(\sigma\sigma^T)^{-1}(\hat{\mu}_t-r_t\mathbf{1}_d)+r_t\Bigr]\,\rmd t. \]
    Now the claim follows.
   \end{proof}
   
  \subsection{Properties of the Value Function}
   
   \begin{corollary}\label{cor:comparison_of_value_functions}
    For any $x_0>0$ it holds $\max\bigl\{V^R(x_0),V^E(x_0)\bigr\}\leqslant V^C(x_0)\leqslant V^F(x_0)$.
   \end{corollary}
   \begin{proof}
    By Proposition~\ref{prop:comparison_of_gammas} we know that $\gamR-\gamC$ is a positive semidefinite matrix for any $t\in[0,T]$. By assumption, $(\sigma\sigma^T)^{-1}$ is positive definite. Hence, the product of these matrices has a non-negative trace, so
    \[ \tr\bigl((\sigma\sigma^T)^{-1}(\gamR-\gamC)\bigr)\geqslant 0. \]
    Therefore,
    \[ \tr\bigl((\sigma\sigma^T)^{-1}\gamR\bigr)\geqslant \tr\bigl((\sigma\sigma^T)^{-1}\gamC\bigr), \]
    which implies by the previous theorem that $V^C(x_0)\geqslant V^R(x_0)$. The same holds for $H=E$ instead of $R$, so $V^C(x_0)\geqslant V^E(x_0)$. Since $\gamF=\mathbf{0}_d$ we also have $V^F(x_0)\geqslant V^C(x_0)$.
   \end{proof}
   
   From Theorem~\ref{thm:asymptotics_for_N_to_infinity} we immediately deduce the following result about the asymptotic behaviour of the value function when the number of expert opinions goes to infinity.
   
   \begin{corollary}\label{cor:asymptotics_of_value_function}
    Let the assumptions of Theorem~\ref{thm:asymptotics_for_N_to_infinity} be fulfilled. Denote the value functions corresponding to the $N$ expert opinions by $V^{E,N}(x_0)$ and $V^{C,N}(x_0)$. Then
    \[ \lim_{N\to\infty} V^{E,N}(x_0) = \lim_{N\to\infty} V^{C,N}(x_0) = V^F(x_0). \]
   \end{corollary}
   \begin{proof}
    Recall from Theorem~\ref{thm:asymptotics_for_N_to_infinity} that
    \[ \lim_{N\to\infty} \bigl\lVert \gamma_u^{E,N} \bigr\rVert = \lim_{N\to\infty} \bigl\lVert \gamma_u^{C,N} \bigr\rVert = 0 \]
    and that $\gamma^F_u = \mathbf{0}_d$ for all $u\in(0,T]$.
    We observe that
    \[ \tr\bigl((\sigma\sigma^T)^{-1}(\Sigma_t+m_tm_t^T-\gamma^H_t)\bigr) \leqslant \tr\bigl((\sigma\sigma^T)^{-1}(\Sigma_t+m_tm_t^T)\bigr) \]
    since both $(\sigma\sigma^T)^{-1}$ and $\gamma^H_t$ are positive semidefinite.
    Using dominated convergence we conclude from the representation of the value function in Theorem~\ref{thm:optimal_value} that $V^{E,N}(x_0)$ and $V^{C,N}(x_0)$ converge to $V^F(x_0)$ when $N$ goes to infinity.
   \end{proof}
   
 \section{Numerical Results}\label{sec:numerical_results}
  
  \subsection{Filters for Various Investors}
   
   After having analyzed in detail the behaviour of the conditional covariance matrices $\gamma^H_t$ for $t\in[0,T]$ we shortly illustrate the development of the filters $\hat{\mu}^H_t$ over time. We have seen that $\hatmuR$ follows a stochastic differential equation. For $\hatmuE$ and $\hatmuC$ there are information dates at which an update of the filter takes place. In the case $H=E$ we have an explicit formula for the development of $\hatmuE$ between these information dates. In the case $H=C$, the filter follows a stochastic differential equation between any two incoming expert opinions.
   
   \begin{example}\label{ex:filters}
    We consider an example for an investment horizon $T$ of one year with $N=10$ equidistant information dates. In our financial market there are three stocks. The expert's covariance matrices are assumed to be constant, i.e.\ $\Gamma_k=\Gamma$ for all $k\in\{0, \dots, N-1\}$. The underlying model parameters are listed in Table~\ref{tab:model_parameters_numerics}. The mean $\delta$ of the drift process $\mu$ is given by the vector $(0.05, 0.10, 0.08)^T\in\R^3$.
    
    In Figure~\ref{fig:filters} one sees one possible realization of the drift process $\mu$ as well as the various filters. The respective first components are plotted in the uppermost subplot, the second components in the middle subplot, and the third components in the lowest subplot. Additionally, the expert opinions $Z_k$ are included in the graphs.
    
    Whereas the filter $\hatmuR$ only takes into account the return observations, $\hatmuE$ only depends on the expert opinions $Z_k$. At each information date the filter $\hat{\mu}^E_{t_k}$ is formed by taking a weighted mean of the former filter and the expert opinion $Z_k$. The combined filter $\hatmuC$ includes both aspects. It takes into account return observations and has update steps at which it takes notice of the information arriving in form of expert opinions.
   \end{example}


   \begin{table}[hb]
    \centering
    \setlength{\tabcolsep}{2pt}
    \begin{tabular}{p{3mm}p{3mm}p{3.8cm}p{4mm}p{3mm}p{2.7cm}}
     \hline\hline
     \addlinespace[2mm]
     $\alpha$ & $=$ & $\begin{pmatrix*}[r] 2\phantom{0} & 1\phantom{0} & -1\phantom{0} \\ 1\phantom{0} & 2\phantom{0} & -1\phantom{0} \\ -1\phantom{0} & -1\phantom{0} & 2\phantom{0} \end{pmatrix*}$ & $\beta$ & $=$ & $\begin{pmatrix} 0.3 & 0.5 & 0.1 \\ 0.5 & 0.2 & 0.2 \\ 0.1 & 0.2 & 0.2 \end{pmatrix}$ \\
     \addlinespace[2mm]
     $\sigma$ & $=$ & $\begin{pmatrix*}[r] 0.30 & 0.08 & 0.05 \\ 0.08 & 0.40 & 0.05 \\ 0.05 & 0.05 & 0.35 \end{pmatrix*}$ & $\Sigma_0$ & $=$ & $\begin{pmatrix} 0.2 & 0.1 & 0.1 \\ 0.1 & 0.3 & 0.1 \\ 0.1 & 0.1 & 0.2 \end{pmatrix}$ \\
     \addlinespace[2mm]
     $\Gamma$ & $=$ & $\begin{pmatrix*}[r] 0.80 & 0.32 & 0.16 \\ 0.32 & 0.72 & 0.24 \\ 0.16 & 0.24 & 0.64 \end{pmatrix*}$ & & & \\ 
     \addlinespace[2mm]
     \hline\hline
    \end{tabular}
    \caption{Model parameters for Example~\ref{ex:filters}}
    \label{tab:model_parameters_numerics}
   \end{table}

   
   \begin{figure}[ht]
    \centering
    \setlength\figureheight{11.2cm}
    \setlength\figurewidth{0.8\textwidth}
    \input{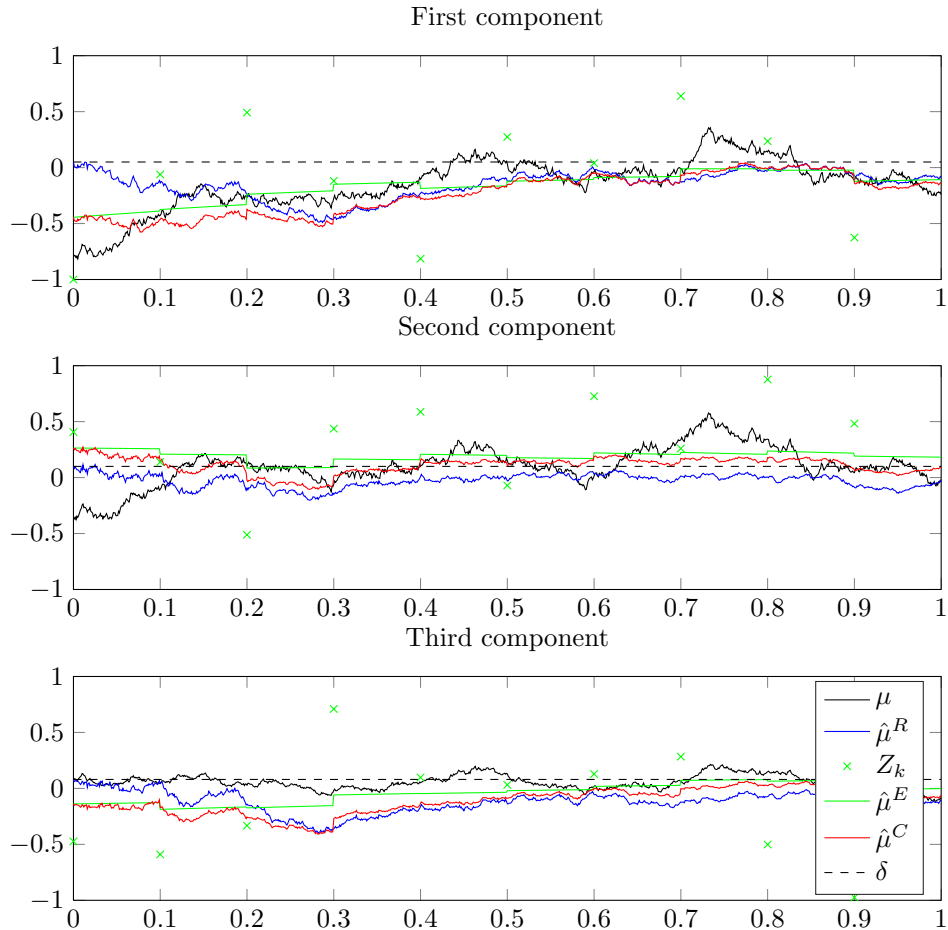}
    \caption{First, second and third component of $\mu$ and the various filters in Example~\ref{ex:filters}}
    \label{fig:filters}
   \end{figure}
   
  \subsection{Analysis of the Value Function}
   
   We can also define the efficiency of the $H$-investor, $H\in\{R,E,C\}$, in the sense of Rogers~\cite{rogers_2001}. We therefore denote by $x_0^H$ the initial capital needed by the $H$-investor to achieve the same expected logarithmic utility of terminal wealth as the $F$-investor starting with an initial capital of $x_0^F=1$. That means, $x_0^H$ is obtained by solving $V^H(x_0^H)=V^F(1)$.
   The value
   \[ \rho^H=\frac{x_0^F}{x_0^H}=\frac{1}{x_0^H} \]
   is then called the \emph{efficiency} of the $H$-investor. It can be shown that in the multi-dimensional case we have
   \[ \rho^H=\exp\biggl(-\frac{1}{2}\int_0^T \tr\bigl((\sigma\sigma^T)^{-1}\gamma^H_t\bigr)\,\rmd t\biggr). \]
   From this representation and our results about the covariance matrices $(\gamma^H_t)_{t\in[0,T]}$ we can deduce that $\max\{\rho^R,\rho^E\}\leqslant\rho^C$. This is intuitive since the $C$-investor can use both return observations and expert opinions for making trading decisions, whereas the $R$-investor and the $E$-investor each only have one of these sources of information at hand.
   
   In Theorem~\ref{thm:asymptotics_for_N_to_infinity} we have seen that when we let the number $N$ of expert opinions go to infinity and the expert covariance matrices $\Gamma^{(N)}_k$ are bounded, the filter covariance matrices $\gamma^{E,N}_t$ and $\gamma^{C,N}_t$ converge to the zero matrix, i.e.\ to $\gamma^F_t$ for each $t\in(0,T]$. By dominated convergence one can conclude
   \[ \lim_{N\to\infty} \rho^{E,N}=\exp\biggl(-\frac{1}{2}\int_0^T \tr\bigl((\sigma\sigma^T)^{-1}\gamma^F_t\bigr)\,\rmd t\biggr) = 1 \]
   and analogously $\lim_{N\to\infty} \rho^{C,N}=1$.
   Hence, in the limit an increasing number of expert opinions yields the highest possible efficiency.
   
   \begin{example}\label{ex:value_function}
    We consider a market with $d=3$ stocks, $m=d$, an investment horizon $T$ of one year and equidistant information dates $t_k=kT/N$ where the expert's covariance matrices are constant, i.e.\ $\Gamma^{(N)}_k=\Gamma$ for all $N\in\N$ and $k\in\{0, \dots, N-1\}$. The model parameters are the same as in Example~\ref{ex:filters}.
    
    For the sake of simplicity we assume that the interest rate $r_t$ of the risk-free bond is zero for all $t\in[0,T]$ and that we start with an initial capital of $x_0=1$. We can now calculate the value function for the different investors and, in the case $H=E$ and $H=C$, different values of $N$. Note that since we have set $r$ to zero and $x_0$ to one, we get as a simpler form of the value function
    \[ V^H(1)= \frac{1}{2}\int_0^T \tr\bigl((\sigma\sigma^T)^{-1}(\Sigma_t+m_tm_t^T-\gamma^H_t)\bigr)\,\rmd t. \]
    We obtain $V^R(1)=0.4503$, $V^F(1)=1.5358$, as well as the values listed in the left part of Table~\ref{tab:value_function_for_various_N_scaled} for $H=E$ and $H=C$. Here, $V^{E,N}(1)$ and $V^{C,N}(1)$ correspond to the situation with $N$ equidistant information dates. Note that $N=0$ yields the special case $V^{C,0}(1)=V^R(1)$ and $V^{E,0}(1)$ is the value function for the investor who has no information at all apart from the model parameters. In that case $\gam{E,0}{t}=\Sigma_t$ for all $t\in[0,T]$. In the last row, we have stated the value function $V^F(1)$ for the fully informed investor.
    
    We can observe that
    \[ \max\bigl\{V^R(1),V^{E,N}(1)\bigr\}\leqslant V^{C,N}(1)\leqslant V^F(1) \]
    for all $N$, a fact that has been proven in Corollary~\ref{cor:comparison_of_value_functions}. Furthermore, the value functions $V^{E,N}(1)$ and $V^{C,N}(1)$ are increasing in $N$. For large values of $N$ the difference between $V^{E,N}(1)$ and $V^{C,N}(1)$ goes to zero and both value functions get close to the value function under full information, $V^F(1)$. This has been proven in Corollary~\ref{cor:asymptotics_of_value_function}. Note that the asserted convergence is rather slow.
    
    The right-hand side of Table~\ref{tab:value_function_for_various_N_scaled} shows the efficiencies of the different investors for the various values of $N$.
   \end{example}


   \begin{table}[ht]
    \centering
    \begin{tabular}{|r|cc|cc|}
     \hline
     $N$ & $V^{E,N}(1)$ & $V^{C,N}(1)$ & $\rho^{E,N}$ & $\rho^{C,N}$\tstrut\bstrut \\ \hline
     0 & 0.0429 & 0.4503 & 22.47\,\% & 33.77\,\%\tstrut \\
     10 & 0.6294 & 0.7414 & 40.40\,\% & 45.18\,\% \\
     100 & 1.1358 & 1.1463 & 67.03\,\% & 67.74\,\% \\
     1000 & 1.4006 & 1.4010 & 87.35\,\% & 87.39\,\% \\
     10000 & 1.4933 & 1.4933 & 95.84\,\% & 95.84\,\% \\ \hline
     $H=F$ &\multicolumn{2}{c|}{1.5358} & \multicolumn{2}{c|}{100.00\,\%}\tstrut \\ \hline
    \end{tabular}
    \caption{Value functions and efficiencies for various $N$ in Example~\ref{ex:value_function}}
    \label{tab:value_function_for_various_N_scaled}
   \end{table}
   
 

\end{document}